\documentclass[11pt]{article}
\usepackage[utf8]{inputenc}

\usepackage{graphicx}
\usepackage{xcolor}
\usepackage{amsmath,amsthm,amssymb}
\usepackage[margin=1in]{geometry}
\usepackage[pdftex,colorlinks=true,linkcolor=blue,citecolor=blue,urlcolor=black]{hyperref}
\usepackage{comment}
\usepackage[normalem]{ulem}
\usepackage{cite}
\usepackage{tcolorbox}
\usepackage{algorithm}
\usepackage{algorithmic}

\DeclareMathOperator*{\argmin}{argmin}

\def\ba{\begin{array}}
\def\ea{\end{array}}
\def\ds{\displaystyle}

\def\0{{\bf 0}}
\def\a{{\bf a}}
\def\b{{\bf b}}

\def\x{{\bf x}}

\newcommand{\ket}[1]{| #1 \rangle}
\newcommand{\bra}[1]{\langle #1|}
\newcommand{\braket}[1]{\langle #1 \rangle}

\newcommand{\update}[1]{{\color{black} #1}}

\newcommand{\be}{\begin{equation}}
\newcommand{\ee}{\end{equation}}
\newcommand{\bea}{\begin{eqnarray}}
\newcommand{\eea}{\end{eqnarray}}
\newcommand{\bes}{\begin{equation*}}
\newcommand{\ees}{\end{equation*}}
\newcommand{\beas}{\begin{eqnarray*}}
\newcommand{\eeas}{\end{eqnarray*}}

\makeatletter
\newtheorem*{rep@theorem}{\rep@title}
\newcommand{\newreptheorem}[2]{%
\newenvironment{rep#1}[1]{%
 \def\rep@title{#2 \ref{##1} (restated)}%
 \begin{rep@theorem}}%
 {\end{rep@theorem}}}
\makeatother

\allowdisplaybreaks

\newtheorem{thm}{Theorem}
\newtheorem*{thm*}{Theorem}

\newtheorem{con}[thm]{Conjecture}
\newtheorem{lem}[thm]{Lemma}
\newtheorem*{lem*}{Lemma}

\newtheorem{prop}[thm]{Proposition}
\newtheorem{defn}[thm]{Definition}

\newreptheorem{thm}{Theorem}
\newreptheorem{lem}{Lemma}
\newreptheorem{cor}{Corollary}



\usepackage{pgfplots}

\usepackage{authblk}
\usepackage[title]{appendix}

\title{Quantum communication complexity of linear regression}

\author[1,2]{Ashley Montanaro\thanks{ashley.montanaro@bristol.ac.uk}}
\author[1]{Changpeng Shao\thanks{changpeng.shao@bristol.ac.uk}}
\affil[1]{School of Mathematics, University of Bristol, UK}
\affil[2]{Phasecraft Ltd. UK}

\date{\today}

\begin{document}

\maketitle

\begin{abstract}

\update{Quantum computers may achieve speedups over their classical counterparts for solving linear algebra problems. However, in some cases -- such as for low-rank matrices -- dequantized algorithms demonstrate that there cannot be an exponential quantum speedup.}
In this work, we show that quantum computers have \update{provable polynomial and exponential} speedups in terms of communication complexity for some fundamental linear algebra problems \update{if there is no restriction on the rank}. We mainly focus on solving linear regression and Hamiltonian simulation. In the quantum case, the task is to prepare the quantum state of the result. To allow for a fair comparison, in the classical case, the task is to sample from the result. We investigate these two problems in two-party and multiparty models, propose near-optimal quantum protocols and prove quantum/classical lower bounds. In this process, we propose an efficient quantum protocol for quantum singular value transformation, which is a powerful technique for designing quantum algorithms. \update{This will be helpful in developing efficient quantum protocols for many other problems.}

\end{abstract}

\section{Introduction}

Quantum computers are designed to solve some problems much faster than classical computers. In particular, quantum computers could be good at solving linear algebra problems. A famous example is the Harrow-Hassidim-Lloyd algorithm for solving linear systems \cite{harrow2009quantum}, whose complexity is only polylog in the dimension. Over the past ten years, quantum linear algebra techniques have been extensively developed especially with the discovery of block-encoding \cite{chakraborty2019power} and quantum singular value transformation \cite{gilyen2019quantum}. For many linear algebra problems, the corresponding quantum algorithms have complexity only polylog in the dimension, which was claimed to be exponentially faster than classical algorithms. However, in 2018, Tang's dequantized algorithm \cite{tang2019quantum} and its development (e.g., \cite{gilyen2022improved,chia2020sampling,shao2021faster,jethwani2020quantum}) showed that quantum computers indeed do not have exponential speedups (in terms of time and query complexity) for many linear algebra problems of low-rank assuming certain data structures.
In this paper, we show that, in the setting of communication complexity \update{and without the low-rank assumption}, provable speedups can be obtained for two fundamental problems: linear regression and Hamiltonian simulation.

\subsection{Our results}

For linear regression problems, we consider two types of models. In the first model, there are only two parties: Alice and Bob. The communication between Alice and Bob can be 1-way or 2-way.
In the second model, there are multiple parties. There is a referee so that the communication is 2-way and only between each party and the referee. We call this model the {\em quantum coordinator model}, as it is a quantum version of the classical coordinator model \cite{vempala2020communication}.

\setlength{\arrayrulewidth}{0.3mm}
{\renewcommand
\arraystretch{1.5}
\begin{table}[t]
\centering
\begin{tabular}{|c|c|c|c|c|} 
 \hline
 & Alice $\rightarrow$ Bob & Bob $\rightarrow$ Alice & Alice $\leftrightarrow$ Bob  \\  \hline
 Quantum & $\widetilde{O}( (\kappa/\gamma)^2 \min(m,n))$ & $\widetilde{O}( (\kappa/\gamma)^2 )$  & $\widetilde{O}(\kappa/\gamma)$ \\
 & $\widetilde{\Omega}(\min(m,n))$ & $\widetilde{\Omega}(\kappa^2+1/\gamma^2)$ &  $\widetilde{\Omega}(\kappa+1/\gamma)$  \\  \hline
 Classical & $\widetilde{O}(m n)$  & $\widetilde{O}(m)$   &  $\widetilde{O}(m)$ \\
 & $\widetilde{\Omega}(\min(m,n))$  & $\Omega(\min(m,n))$ & $\Omega(\min(m,n))$ \\ \hline
 Quantum speedups & at most quadratic  &  can be exponential &  can be exponential \\ \hline
\end{tabular}
\caption{Comparison of quantum and classical communication complexity for solving the linear regression problem $\x_{\rm opt}=\arg\min_{\x}\|A\x-\b\|$. Alice has $A$ and Bob has $\b$. All the lower bounds hold even if $m=n, \kappa = O(1), 1/\gamma = O(1)$. \update{ Here $\kappa$ is the condition number of $A$, i.e., the ratio of the maximal and minimal nonzero singular values of $A$, and $\gamma = {\|A \x_{\rm opt}\|}/{\|\b\|}$ which describes the overlap of $\b$ in the column space of $A$. With $\widetilde{O}, \widetilde{\Omega}$, we ignore all the polylog factors in the input size, the condition number and the accuracy}. The arrow means the direction of communication. The results are presented rigorously in Theorems \ref{thm:quantum protocol}, \ref{thm:Quantum lower bound} and \ref{thm:Classical lower bound}.}
\label{comparison: Alice Bob case}
\end{table}
}

\subsubsection{Alice-Bob model}

The setting here is that Alice has a matrix $A\in \mathbb{R}^{m\times n}$ and Bob has a vector $\b\in \mathbb{R}^{m}$. Their goal is to solve the linear regression problem $\arg\min_{\x}\|A\x-\b\|$ together using as little communication as possible. More precisely,

\begin{itemize}
\item 
In the quantum case, their goal is to prepare a quantum state $\ket{\tilde{\x}_{\rm opt}}$ that is $\varepsilon$ close to $\ket{\x_{\rm opt}} = \ket{A^+\b}$ {in trace distance}, where $A^+$ is the pseudoinverse of $A$,\footnote{\update{There are many equivalent ways to define pseudoinverse \cite{golub2013matrix}. Here we recall the one based on singular value decomposition (SVD). If $A=\sum_{i=1}^r \sigma_i \ket{u_i} \bra{v_i}$ is SVD of $A$, where $r=\text{Rank}(A)$, then $A^+=\sum_{i=1}^r \sigma_i^{-1} \ket{v_i} \bra{u_i}$.}} and $\ket{A^+\b}$ denotes the normalised state corresponding to $A^+\b$. This corresponds to an optimal solution to the problem of minimising $\|A\x-\b\|$. 
\item 
In the classical case, the goal is to sample from a distribution $\widetilde{\mathcal{P}}(\ket{\x_{\rm opt}})$ that is $\varepsilon$ close to the distribution $\mathcal{P}(\ket{\x_{\rm opt}})$ defined by $\ket{\x_{\rm opt}}$ in the total variation distance, i.e., $\|\widetilde{\mathcal{P}}(\ket{\x_{\rm opt}}) - \mathcal{P}(\ket{\x_{\rm opt}})\|_{1} \leq \varepsilon$.
\end{itemize}

The problem solved by the quantum computer is at least as hard as the problem solved by the classical computer. If the communication is 2-way, then Alice and Bob can send quantum/classical information to each other. If the communication is 1-way, then only Alice or Bob can send quantum/classical information to the other party. 
In communication complexity, we are interested in the minimal amount of communication (which is described by the number of qubits or bits) the parties used to achieve their goal. Since our main focus is on the quantum speedup with respect to dimension, throughout we assume that matrix entries are specified with $O(\log (mn))$ bits. Our main results are summarised in Table \ref{comparison: Alice Bob case}. 

From Table \ref{comparison: Alice Bob case}, we can see that
\begin{itemize}
\item If the communication is 1-way from Alice to Bob, then the quantum speedup is at most quadratic. The quantum protocol in this case is optimal if the linear regression problem is well-conditioned. \update{The quadratic speedup here is not related to Grover's algorithm or more generally the amplitude amplification. 
A key point here is that Bob holds too little information about the linear regression they aim to solve. When the communication is 1-way from Alice to Bob, then even in the quantum case, Alice still needs to send a lot of information about the matrix to Bob.}

\item If the communication is 1-way from Bob to Alice or 2-way, then the quantum speedup is exponential if the linear regression problem is well-conditioned. The quantum protocols in these two cases are optimal up to a polylogarithmic factor.
\end{itemize}

\subsubsection{Coordinator model}

In the second model, we consider a more general setting. Now we suppose there are $r$ parties $P_0,\ldots,P_{r-1}$. The party $P_i$ holds a matrix $A_i \in \mathbb{R}^{d_i\times n}$ and a vector $\b_i \in \mathbb{R}^{d_i}$. Their goal is to solve the linear regression problem
\be
\label{problem 2}
\argmin_\x \|A\x-\b\|, \quad \text{ where }
A = \begin{pmatrix}
    A_0 \\
    \vdots \\
    A_{r-1}
\end{pmatrix}, ~ \b=
\begin{pmatrix}
    \b_0 \\
    \vdots \\
    \b_{r-1}
\end{pmatrix} .
\ee
In this case, there is a referee and every party can only communicate with the referee. Their goal is similar, i.e., up to certain errors, outputting the quantum state of the optimal solution quantumly or sampling from the optimal solution classically by one party or by the referee. In the classical case, Vempala, Wang, and Woodruff studied the problem (\ref{problem 2}) with the goal of outputting the whole vector of the optimal solution 
\cite{vempala2020communication}. A near-optimal protocol of complexity $O(r n^2)$ was given. The lower bound is $\Omega(r n+n^2)$. Here our goal is different from theirs.

In the quantum case, if we consider the problem in the simultaneous message passing (SMP) model (in which the referee is not allowed to send information to the parties), then we show that $\Omega(n)$ qubits of communication are required for any quantum protocols to solve (\ref{problem 2}). Because of this and also inspired by the classical coordinator model \cite{vempala2020communication}, we assume that the communication is 2-way between each party and the referee (also known as the coordinator classically). We call this the quantum coordinator model. Our main result is summarized in Table \ref{comparison2: Alice Bob case}.

\setlength{\arrayrulewidth}{0.3mm}
{\renewcommand
\arraystretch{1.5}
\begin{table}[t]
\centering
\begin{tabular}{|c|c|c|c|c|} 
 \hline
 Quantum                        & Classical \\  \hline
 $\widetilde{O}(r^{1.5}\kappa/\gamma)$ & $O(rn^2)$ \\   
 $\Omega(r\kappa)$              & $\Omega(rn)$ \\ \hline
\end{tabular}
\caption{Comparison of quantum and classical communication complexity for solving (\ref{problem 2}) in the coordinator model. Here $\kappa$ is the condition number of $A$ defined in (\ref{problem 2}) and $\gamma=\|A\x_{\rm opt}\|/\|\b\|$. The results are presented rigorously in Theorems \ref{thm:solve the general linear regression} and \ref{thm stronger lower bound}.}
\label{comparison2: Alice Bob case}
\end{table}
}


From Table \ref{comparison2: Alice Bob case}, we have
\begin{itemize}
\item 
The quantum protocol has an optimal dependence on the condition number. Also, when $r=O({\rm polylog}(mn))$, quantum computers are exponentially faster than classical computers for well-conditioned linear regression problems.
\item 
Similar to the result of \cite{vempala2020communication}, our result shows that it is hard for a classical computer even for a weak task of solving linear regressions.
\end{itemize}

\subsection{Summary of techniques}

In the Alice-Bob model, the quantum protocols are straightforward. For example, if the communication is 1-way from Bob to Alice, then Bob just sends the quantum state of $\b$ to Alice, who applies $A^{+}$ to this quantum state and performs some measurements. The interesting part is the lower bound analysis, which is based on the hardness of the disjointness problem and the index problem \cite{klauck2000quantum, kremer1999randomized,buhrman2001communication,razborov2003quantum}. In the disjointness problem, Alice and Bob respectively have a subset $x$ and $y$ of $[n]$, their goal is to determine if $x\cap y \neq \emptyset$. This problem can be reduced to a linear regression problem by constructing a diagonal matrix $D$ from $x$ and a vector $\b$ from $y$ as follows: We set $D_i = 1$ if $i\in x$ and $D_i=1/\varepsilon$ if $i\notin x$ for some small $\varepsilon$. Similarly, we define $b_i = 1$ if $i\in y$ and $b_i=\varepsilon$ if $i\notin y$. It is not hard to see that the indices in $x\cap y$ have large amplitudes in the quantum state $\ket{D^{-1}\b}$. The index problem is used in a similar way. This is indeed the main idea of our quantum/classical lower bound analysis. In different settings, we construct appropriate diagonal matrices and vectors.

The quantum protocol for solving (\ref{problem 2}) is based on quantum singular value transformation (QSVT) \cite{gilyen2019quantum}, which is a useful technique for designing quantum algorithms (e.g., see the survey paper \cite{martyn2021grand}). In the quantum coordinator model,
we show that QSVT is still applicable and efficient (see Proposition \ref{prop:QSVT in CC}). Unlike time and query complexity, the communication complexity of implementing QSVT can be estimated precisely. For example, if we apply QSVT to solve the linear regression problem (\ref{problem 2}), then the time complexity is $\widetilde{O}((T_A+T_b)\alpha/\gamma\sigma_{\min})$ \cite[Corollary 31]{chakraborty2019power}, where $T_A$ is the complexity of constructing the block-encoding of $A$ so that $A/\alpha$ is the top-left corner of a unitary, $T_b$ is the complexity of preparing the quantum state of $\b$, and $\sigma_{\min}$ is the minimal singular value of $A$. Regarding the communication complexity, we can show that $T_A=T_b=O(r\log n)$ and $\alpha = O(\sqrt{r}\|A\|)$. Here $T_A, T_b$ should be understood as communication complexity. We can also show that the above formula for time complexity is still true so that the communication complexity is $\widetilde{O}(r^{1.5}\kappa/\gamma)$. This is exactly the result we stated in Table \ref{comparison2: Alice Bob case}.
Since QSVT is a powerful technique in designing quantum algorithms, we believe that for many other linear algebra problems, quantum computers still have provable speedups in terms of communication complexity. 

Indeed, as an application, we show that quantum computers achieve provable speedups for Hamiltonian simulation. In the Hamiltonian simulation problem, we
suppose that the party $P_i$ holds a Hamiltonian $H_i$ of dimension $n$, the referee holds a quantum state $\ket{\psi}$, and their goal is to prepare the state $e^{i(H_0+\cdots+H_{r-1})t} \ket{\psi}$ quantumly or sample from it classically. In Propositions \ref{prop:classical-HS} and \ref{prop:SMP HS}, we show that the quantum communication complexity of this Hamiltonian simulation problem is $\widetilde{O}(r|t|\sum_{i=0}^{r-1}  \|H_i\|)$ and
$\Omega(n)$ bits communication are required for any classical protocols to sample from $e^{i(H_{0}+\cdots+H_{r-1})t} \ket{\psi}$.


\subsection{Related work}

The problem studied in this paper is partially inspired by \cite{vempala2020communication}, in which Vempala, Wang and Woodruff studied the classical communication complexity of solving linear regression (and many other optimization problems) and showed that the naive protocol (of sending the whole information to others) is close to optimal. Their goal is to output a vector solution, while our goal is to sample from the solution. Recently, Tang et al. \cite{Hao2022} studied quantum communication complexity of solving linear regression problems. Their focus was on the Alice-Bob model and their goal was to output an approximate vector solution. The complexity they proved is $O(n\kappa/\varepsilon)$, where $n$ is the dimension, $\kappa$ is the condition number and $\varepsilon$ is the precision.
Regarding quantum communication complexity for sampling problems, Ambainis et al.
\cite{ambainis2003quantum} exhibited an exponential gap between the quantum and classical communication required for a sampling problem related to disjointness.
In \cite{montanaro2019quantum}, Montanaro showed an exponential gap between the quantum and classical communication for a distributed variant of the Fourier sampling
problem.
Another linear algebra problem that shows quantum computers are exponentially better than classical computers is the vector-in-subspace problem studied by Raz
in \cite{raz1999exponential}. This problem is closely related to the two-party linear regression problem we study in the case where communication is from Bob to Alice, and where the matrix $A$ is unitary. It was proved in \cite{Klartag-Regev} that the one-way quantum
protocol is exponentially better than any classical protocol, even if the latter is allowed bounded error and two-way communication.
When restricted to finite fields, Sun and Wang \cite{sun2012randomized} studied quantum/classical communication complexity of matrix singularity and determinant computation problems. Their results suggest that there is no exponential quantum speedup for those problems in terms of dimension.

\section{Preliminaries}

\subsection{Communication complexity}

Communication complexity has been studied extensively in the field of classical and quantum computing \cite{yao1979some,yao1993quantum,de2002quantum,klauck2000quantum,buhrman2010nonlocality}. It usually deals with the following type of problem. Suppose there are two separated parties: Alice and Bob. Alice receives some input $x\in X$ and Bob receives some input $y \in Y$. Their goal is to compute $f(x,y)$ together using as little communication as possible. We usually assume that Alice and Bob have unlimited computational power so that they can perform any computation as efficiently as they want. The measure of complexity used is the amount of communication required to solve the problem. All other non-communication operations are treated as free.
A protocol is an algorithm where first Alice does some individual computation, and then sends a quantum/classical message to Bob, then Bob does some individual computation and sends a quantum/classical message to Alice, etc. In the end, one of the parties outputs some value that should be $f(x,y)$. 
A quantum message usually refers to a quantum state. The cost of sending a quantum state is described by the number of qubits the quantum state occupies. In contrast, the cost of sending a classical message is the number of bits the message uses. 

The cost of a protocol is the total number of bits/qubits communicated on the worst-case input. We are more concerned about the minimal amount of communication they need.  A {\em deterministic protocol} for $f$ always has
to output the right value $f(x,y)$ for all $(x,y)$. In a {\em bounded-error protocol}, the protocol has to output the right value $f(x,y)$ with probability at least $2/3$ for all $(x,y)$. \update{In the randomised model, Alice and Bob share an unlimited supply of uniformly random bits, which they can use in deciding what messages to send. Also, an error is allowed in randomised protocols, which means the output of the protocol is correct with probability at least 2/3. In this work, in the classical case, we will focus on randomised communication complexity. In the quantum case, we focus on bounded error communication complexity.}

If only one party (say Alice) can send information to another party, then this is known as {\em 1-way communication model}. Otherwise, it is a {\em 2-way communication model}. More generally, there can be multiple parties. In this case, there are many ways to define the models. For example, in the {\em simultaneous message passing} (SMP) model, there is a referee so that each party can only communicate with the referee. If the communication is 2-way, then it is also known as {\em coordinator model} classically \cite{vempala2020communication}. 

Regarding communication complexity, two fundamental problems are the {\em index problem} and the {\em disjointness problem}. These problems are well-studied classically and quantumly. They also play significant roles in this paper for the lower bounds estimation.
In the index problem, Alice has a bit string $(x_1,\ldots,x_n)\in \{0,1\}^n$ and Bob has an index $j\in[n]$. The goal is to determine $x_j$. This problem is trivial if Bob can send information to Alice. Namely, Bob just sends the index $j$ to Alice, and Alice outputs $x_j$. The index problem is hard if the communication is 1-way from Alice to Bob. In the disjointness problem, Alice has a bit string $(x_1,\ldots,x_n)\in \{0,1\}^n$ and Bob has another bit string $(y_1,\ldots,y_n)\in \{0,1\}^n$. The goal is to determine if there is an index $j$ such that $x_j=y_j=1$. We summarize the known results for these two problems into the following proposition.



\update{

\begin{prop}
\label{prop1}
We have the following.

\begin{enumerate}
\item The classical 1-way communication complexity of the index problem is $\Theta(n)$ \cite{kremer1999randomized}.

\item The quantum 1-way communication complexity of the index problem is $\Theta(n)$ \cite{buhrman2001communication}.

\item The classical 2-way communication complexity of the disjointness problem is $\Theta(n)$ \cite{Kalyanasundaram1992,razborov1990distributional}.

\item The quantum 2-way communication complexity of the disjointness problem is $\Theta(\sqrt{n})$ \cite{aaronson2003quantum,razborov2003quantum}.
\end{enumerate}

\end{prop} 

Another known result we will use is called Distributed Fourier Sampling \cite{montanaro2019quantum}. In this problem, Alice is given a function $f:\{0,1\}^n \rightarrow \{\pm 1\}$, Bob is given a function $g:\{0,1\}^n \rightarrow \{\pm 1\}$, their task is for one party (say Bob) to approximately sample from the distribution $P_{fg}$ on $n$-bit strings $s$ where
\[
P_{fg}(s) = \left( \frac{1}{2^n} \sum_{x\in\{0,1\}^n} (-1)^{s\cdot x} f(x) g(x)  \right)^2,
\]
and $s\cdot x = \sum_i s_ix_i$. That is, Bob must output a sample from any distribution $\widetilde{P}_{fg}$ such that $\|\widetilde{P}_{fg} - P_{fg}\|_1 \leq \varepsilon$ for some constant inaccuracy $\varepsilon$. 

\begin{prop}[Theorem 1 of \cite{montanaro2019quantum}]
\label{prop2}
There exist universal constants $\varepsilon, \gamma > 0$ such that, for sufficiently large $n$, any 2-way classical communication protocol for Distributed Fourier Sampling with shared randomness and inaccuracy $\varepsilon$ must communicate at least $\gamma 2^n$ bits.
\end{prop} 

The last result will be used in this work is the multi-player set-disjointness problem \cite{phillips2012lower}. In this problem, player $P_j$ holds a subset $T_j \in [n]$, where $j\in\{1,\ldots,k\}$, and their goal is to determine if there is a $j\in\{2,\ldots,k\}$ such that $T_1\cap T_j \neq \emptyset$. This problem was studied in the coordinator model. 

\begin{prop}[Theorem 3.1 of \cite{phillips2012lower}]
\label{prop3}
Assume that $n\geq 3200 k$, then the communication complexity of the multi-player set-disjointness problem is lower bounded by $\Omega(kn)$ in the coordinator model.
\end{prop} 

}

\subsection{Block-encoding}

Our quantum protocols in the multiparty model are based on quantum singular value transformation (QSVT) \cite{gilyen2019quantum}. A key ingredient of using QSVT is block-encoding. For convenience, in this section,  we list some results about block-encoding that will be used in our quantum protocols. 

\begin{defn}[Block-encoding, c.f. Definition 24 of \cite{gilyen2019quantum}]
Suppose that $A$ is an $s$-qubit operator, $\alpha\geq \|A\|, \varepsilon \in \mathbb{R}^{>0}$ and $q\in \mathbb{N}$, then we say that the $(s+q)$-qubit unitary $U_A$ is an $(\alpha, q, \varepsilon)$ block-encoding of $A$, if
\be
\|A - \alpha (\bra{0}^{\otimes q}\otimes I) U_A (\ket{0}^{\otimes q}\otimes I) \| \leq \varepsilon,
\ee
where $\|\cdot\|$ is the operator norm. In matrix form
\be
U_A = \begin{pmatrix}
A/\alpha & \cdot \\
\cdot & \cdot \\
\end{pmatrix}
\ee
up to an error $\varepsilon$.
\end{defn}

In quantum computing, we hope $\alpha$ is as small as possible. It is obvious that the optimal choice is $\alpha = \|A\|$. In the model of communication complexity, we assume that each party has unlimited computational power, so each party can first compute the singular value decomposition (SVD) of $A$ and use it to construct the block-encoding with $\alpha = \|A\|$. We can even assume that $\varepsilon=0$. More precisely, if $A=UDV^T$ is the SVD of $A$, then 
\be
\label{block-encoding by SVD}
\begin{pmatrix}\vspace{.2cm}
U & \\
  & I
\end{pmatrix}
\begin{pmatrix} \vspace{.2cm}
    D/\|A\| & \sqrt{I - D^2/\|A\|^2} \\
    \sqrt{I - D^2/\|A\|^2} & -D/\|A\|
\end{pmatrix}
\begin{pmatrix}\vspace{.2cm}
V^T & \\
  & I
\end{pmatrix}
\ee
is an $(\|A\|, 1, 0)$ block-encoding of $A$. In the above, we implicitly assumed that $A$ is square, otherwise we can add some zero rows or columns. In this paper, we will always use this block-encoding.

The following result is a direct application of the technique of linear combination of unitaries \cite{berry2015simulating}.

\begin{lem}
\label{lem:block-encoding construction}
For each $i\in\{0,1,\ldots,r-1\}$, let $U_i$ be an $(\alpha_i, q,0)$ block-encoding of $A_i \in \mathbb{R}^{d_i\times n}$, \update{where $\alpha_i>0$}. Let $V$ be a unitary such that 
\be \label{lcu:eq}
V \ket{0} = \frac{1}{\alpha} \sum_{i=0}^{r-1} \alpha_i \ket{i},
\ee
where $\alpha = \sqrt{\sum_i\alpha_i^2}$.
Then
\be
\label{block-encoding-unitary}
({\rm SWAP}_{1,2}\otimes I_n)\left(\sum_{i=0}^{r-1} \ket{i} \bra{i}  \otimes U_i\right) (V\otimes I_{2^q} \otimes I_n) 
\ee
is an $(\alpha, q+\log r,0)$ block-encoding of 
\[A := \begin{pmatrix}
A_0 \\
\vdots \\
A_{r-1}
\end{pmatrix},\]
where ${\rm SWAP}_{1,2}$ swaps the first two registers containing $\log r$ and $q$ qubits respectively.
\end{lem}

\begin{proof}
Denote the unitary (\ref{block-encoding-unitary}) as $W$.
We can check that for any state $\ket{\psi}$
\beas
W \ket{0}^{\otimes \log r} \ket{0}^{\otimes q} \ket{\psi} &=&  \frac{1}{\alpha} \sum_{i=0}^{r-1} \ket{0}^{\otimes q} \otimes \ket{i} \otimes A_i \ket{\psi} + \text{orthogonal terms} \\
&=&  \frac{1}{\alpha} \ket{0}^{\otimes q} \otimes A \ket{\psi} + \text{orthogonal terms}.
\eeas
This means that $W$ is a block-encoding of $A$.
\end{proof}

Similarly, we have the following result.

\begin{lem}
\label{lem:block-encoding construction2}
For each $i\in\{0,1,\ldots,r-1\}$, let $U_i$ be an $(\alpha_i, q,0)$ block-encoding of $A_i \in \mathbb{R}^{m\times n}$, \update{where $\alpha_i>0$}. Let $V$ be a unitary such that 
\be \label{lcu:eq2}
V \ket{0} = \frac{1}{\sqrt{\alpha}} \sum_{i=0}^{r-1} \sqrt{\alpha_i} \ket{i},
\ee
where $\alpha = \sum_i\alpha_i$.
Then
\be
\label{block-encoding-unitary2}
(V^\dag \otimes I_{2^q} \otimes I_n) \left(\sum_{i=0}^{r-1} \ket{i} \bra{i}  \otimes U_i\right) (V\otimes I_{2^q} \otimes I_n) 
\ee
is an $(\alpha, q+\log r,0)$ block-encoding of $A_0+\cdots+A_{r-1}$.
\end{lem}

{\bf Notation.} For any matrix $A$, with $A^{+}$, we always mean the  pseudoinverse of $A$. 
The operator norm of $A$ is denoted as $\|A\|$, which is the largest singular value. The Frobenius norm is denoted as $\|A\|_F$. It is defined as the square root of the sum of the absolute squares of the elements of $A$.  The condition number $\kappa$ of matrix $A$ is defined as the ratio of the largest singular value and the smallest nonzero singular value.

\section{Two parties}

In this section, we focus on the Alice-Bob model. In this model, Alice receives a matrix $A \in \mathbb{R}^{m\times n}$ and Bob receives a vector $\b \in \mathbb{R}^{m}$, and their goal is to solve the linear regression problem
\be
\label{linear regression problem:2 parties}
\argmin_{\x \in \mathbb{R}^n} \quad \|A\x-\b\|^2
\ee
through 1-way or 2-way communication.\footnote{Throughout, by an optimal solution of a linear regression problem $\argmin_{\x \in \mathbb{R}^n} \|A\x-\b\|^2$, we always mean the solution $A^+\b$, where $A^+$ is the pseudoinverse of $A$. It is possible that a linear regression can have infinitely many solutions; however, the solution with minimum norm is unique, which is  $A^+\b$ \cite{planitz19793}.} In the quantum case, their goal is to output the quantum state of the optimal solution. In the classical case, their goal is to sample from the optimal solution.\footnote{Classically, the task of outputting a vector solution has been studied in \cite{vempala2020communication}. In this paper, we show that sampling is also hard for classical computers. Moreover, this makes the quantum speedup in solving linear regression problems more convincing.} This sampling task is partially inspired by the quantum-inspired classical algorithms \cite{tang2019quantum}. Also, sampling is a natural application of measuring the quantum state of the solution. But outputting a quantum state could be a much harder task than sampling because we can perform many other operations on a quantum state.

If the communication is 2-way, then either Alice or Bob can output the result (i.e., the quantum state or a sample). If the communication is 1-way, then only Alice or Bob can output the result depending on the direction of the communication. 

Since our main goal is to demonstrate the quantum advantage in terms of dimension, we assume that entries of $A$ and $\b$ can be specified using $O(\log(mn))$ bits for simplicity. \update{In this work, many of the quantum protocols require the communication of the operator norm $\|A\|$ and the 2-norm $\|\b\|$. When entries are specified by $O(\log(mn))$ bits, these two norms are also specified by $O(\log(mn))$ bits. This is enough for us since we indeed only need a good upper bound of these quantities in our quantum protocols.
For the sampling task, the norm of $\b$ is unimportant, so in this paper, we do not assume that $\b$ is a unit vector even if sometimes we use the notation $\ket{\b}$.}

The optimal solution of (\ref{linear regression problem:2 parties}) is $\x_{\rm opt}=A^+\b$. Since Alice can compute $A^+$ in advance, the linear regression problem is indeed equivalent to the matrix-vector multiplication problem. 
We define
\be \label{def:gamma}
\gamma:= \frac{\|A \x_{\rm opt}\|}{\|\b\|} ,
\ee
which describes the fraction of the norm of $\b$ that lies in the column space of $A$. 
It has an interesting geometric explanation. Namely, it is the cosine of the angle between $\b$ and the column space of $A$. If $\gamma=1$, then the linear system $A\x=\b$ is consistent, and $\x_{\rm opt}$ is the solution with the minimum norm. If $\gamma = 0$, then $\x_{\rm opt} = 0$.

\subsection{The quantum protocols}

In this section, we present the quantum protocols for solving linear regression problems in the 1-way and 2-way models.

{\bf Case 1 (1-way). Only Bob can send quantum information to Alice.}

In this case, it is Alice that needs to output the quantum state of the solution. The quantum protocol is straightforward. Namely, Bob just sends the quantum state $\ket{\b}$ to Alice, then Alice applies $A^{+}$ to $\ket{\b}$. To this end, Alice constructs a block-encoding of $A^{+}$, i.e., constructs a unitary $U_{A}$ using the SVD of $A$ such that
\[
U_A = \begin{pmatrix}
A^{+}/\|A^{+}\| & \cdot \\
\cdot & \cdot
\end{pmatrix}.
\]
Then she applies $U_{A}$ to $\ket{0}\ket{\b}$ and obtains
\be \label{case 1:solution state}
\frac{1}{\|A^{+}\|} \ket{0}  \otimes A^{+}\ket{\b} + \ket{0}^\bot
=\frac{\|A^{+}\ket{\b}\|}{\|A^{+}\|} \ket{0}  \otimes \ket{\x_{\rm opt}} + \ket{0}^\bot,
\ee
where $\ket{0}^\bot$ refers to some orthogonal terms. Now Alice can measure the first register. If she receives $\ket{0}$, then the post-selected state is $\ket{\x_{\rm opt}}$.
The success probability \update{of Alice seeing $\ket{0}$ in the first register} is
\[
\frac{\|A^{+} \ket{\b}\|^2}{\|A^{+}\|^2}.
\]
Since the communication is 1-way, Alice cannot use the quantum amplitude amplification technique. Thus for Alice to obtain a copy of the state $\ket{\x_{\rm opt}}$, they need to repeat the above procedure $O(\|A^{+}\|^2/\|A^{+} \ket{\b}\|^2)$ times, i.e.,  Bob sends $O(\|A^{+}\|^2/\|A^{+} \ket{\b}\|^2)$ copies of the state $\ket{\b}$ to Alice. Therefore, the quantum communication complexity is $O((\log m)\|A^{+}\|^2/\|A^{+} \ket{\b}\|^2)$. \update{Usually, Bob does not know $\|A^{+}\|^2/\|A^{+} \ket{\b}\|^2$ exactly, which depends on $A$ and $\b$. Here we assume that Bob knows a good upper bound on it so he knows how many copies need to be sent to Alice.}

To obtain a clear intuition about the complexity, we can bound the complexity in terms of $\kappa$ (the condition number of $A$) and $\gamma$ (defined in (\ref{def:gamma})).
Suppose the SVD of $A = \sum_{i=1}^r \sigma_i \ket{u_i} \bra{v_i}$ and $\b = \sum_{i=1}^m \beta_i \ket{u_i}$, where $r={\rm Rank}(A)$ and $\sigma_1\geq \cdots\geq \sigma_r>0$. Then $\|A^+\| = 1/\sigma_r$ and
\be
\label{lower bound of b}
\|A^{+} \b\|^2 = \sum_{i=1}^r \frac{|\beta_i|^2}{\sigma_i^2}
\geq \frac{1}{\sigma_1^2} \sum_{i=1}^r |\beta_i|^2 = \frac{1}{\sigma_1^2} \frac{\|A \x_{\rm opt}\|^2}{\|\b\|^2} = \frac{\gamma^2}{\sigma_1^2}.
\ee
So the communication complexity is bounded by $O((\log m)\kappa^2/\gamma^2)$.

It is possible that $\|A^{+} \ket{\b}\|$ (or $\gamma$) is very small or even zero. This happens when $\b$ is far from the column space of $A$. In this case, there is a small success probability to obtain the solution state by measuring the first register of the state (\ref{case 1:solution state}). If after $O(m\log (mn))$ measurements Alice still does not receive $\ket{0}$, then this means that the success probability is small. When this happens, Bob can just send the whole vector to Alice. This costs $O(m\log(mn))$.



{\bf Case 2 (1-way). Only Alice can send quantum information to Bob.}

In this case, Bob needs to output the quantum state of the solution. The quantum protocol reads as follows. Alice computes $A^{+}$ and prepares the quantum state of $A^{+}$:
\[
\ket{A^{+}} := \frac{1}{\|A^{+}\|_F} \sum_{i\in[n], j \in[m]} (A^{+})_{ij} \ket{i} \ket{j} = 
\frac{1}{\|A^{+}\|_F} \sum_{j \in[m]} A^{+}\ket{j} \otimes \ket{j}.
\]
Then she sends $\ket{A^{+}}$ to Bob. Since Bob has the vector $\b$, he can construct a unitary $U_b$ such that $U_b\ket{0} = \ket{\bar{\b}}$, where $\bar{\b}$ is the complex conjugate of $\b$. Now he applies $U_b^\dag$ to the second register $\ket{j}$ of $\ket{A^{+}}$. The resulting state is
\[
\frac{1}{\|A^{+}\|_F} \sum_{j \in[m]} A^{+}\ket{j} \otimes U_b^\dag \ket{j}
=
\frac{1}{\|A^{+}\|_F} A^{+}\ket{\b} \otimes \ket{0} + \ket{0}^\bot
=
\frac{\|A^{+}\ket{\b}\|}{\|A^{+}\|_F} \ket{\x_{\rm opt}} \otimes \ket{0} + \ket{0}^\bot.
\]
\update{
Regarding the first equality, note that $U_b\ket{0} = \ket{\bar{\b}}$, i.e., the first column of $U_b$ is $\ket{\bar{\b}}$, so we have $U_b^\dag \ket{j} = b_j \ket{0} + \ket{0}^\bot$. This means
\[
\sum_{j \in[m]} A^{+}\ket{j} \otimes U_b^\dag \ket{j}
=\sum_{j \in[m]} b_j A^{+}\ket{j}  \otimes  \ket{0} + \ket{0}^\bot
=A^{+}\ket{\b} \otimes \ket{0} + \ket{0}^\bot.
\]
}
The success probability of obtaining $\ket{\x_{\rm opt}}$ is
\[
\frac{\|A^{+}\ket{\b}\|^2}{\|A^{+}\|_F^2}.
\]
This means that Alice needs to send $O(\|A^{+}\|_F^2/\|A^{+}\ket{\b}\|^2)$ copies of the state $\ket{A^{+}}$ to Bob. So the total number of qubits in communication is $O((\log(mn))\|A^{+}\|_F^2/\|A^{+}\ket{\b}\|^2)$. Note that 
\[
\|A^{+}\|_F^2=\sum_{i=1}^r \frac{1}{\sigma_i^2} \leq \frac{\min(m,n)}{\sigma_r^2},
\]
together with (\ref{lower bound of b}),
the communication complexity is bounded by
$
O((\log(mn))\min(m,n)\kappa^2/\gamma^2).
$

Similar to the discussion in case 1, if $\gamma$ is too small, then Alice can just send the whole matrix to Bob, which uses $O(mn\log(mn))$ qubits in communication.

{\bf Case 3 (2-way). Alice and Bob can send quantum information to each other.}

In this case, either one can output the solution state. They still use the protocol designed in the first case. Since it is 2-way quantum communication, they can use the quantum amplitude amplification technique \cite{brassard2002quantum}. More precisely, denote the state (\ref{case 1:solution state}) as $\ket{\psi}=U_A \ket{0}\ket{\b}$. To apply the quantum amplitude amplification, the main obstacle for Alice and Bob is to perform the reflection 
\[
2\ket{\psi} \bra{\psi} - I
=
U_A (2\ket{0}\ket{\b} \bra{0}\bra{\b} - I) U_A^\dag.
\]
They can achieve this as follows.
For any state, say in Bob's hand, if he wants to apply the reflection he can first send the state to Alice who applies $U_A^\dag$ to it. Then Alice sends the new state back to Bob who can apply the reflection $2\ket{0}\ket{\b} \bra{0}\bra{\b} - I$. After that he sends the state to Alice again and asks her to apply $U_A$ to the state. Finally, Alice sends the resulting state to Bob.

Therefore, the total number of copies Bob needs to send to Alice is $O(\|A^{+}\|/\|A^{+} \ket{\b}\|)$. This means that the quantum communication complexity is $O((\log m)\|A^{+}\|/\|A^{+} \ket{\b}\|)$, which is bounded from above by $O((\log m)\kappa /\gamma)$. In summary, we have the following result.

\begin{thm}
\label{thm:quantum protocol}
Suppose Alice has a matrix $A \in \mathbb{R}^{m\times n}$ and Bob has a vector $\b \in \mathbb{R}^{m}$. The quantum communication complexity of outputting $\ket{A^{+}\b}$ is
\begin{enumerate}
\item $O(\min\{(\log m)\|A^{+}\|^2 \|\b\|^2/\|A^{+} \b\|^2 , m\log(mn)\})$ if the communication is 1-way from Bob to Alice.
\item $O(\min\{(\log(mn))\|A^{+}\|_F^2\|\b\|^2/\|A^{+} \b\|^2,mn\log(mn)\})$ if the communication is 1-way from Alice to Bob.
\item $O(\min\{(\log m)\|A^{+}\|\|\b\|/\|A^{+} \b\|,m\log(mn)\})$ if the communication is 2-way.
\end{enumerate}
\end{thm}


As a direct corollary, if $A$ is well-conditioned and $\b$ lies in the column space of $A$ (e.g., $A$ is unitary), then the communication complexity in case 1 and case 3 is $O(\log m)$, and the communication complexity in case 2 is $O(\min(m,n)\log (mn))$ \update{since $\|A^{+}\|_F^2\|\b\|^2/\|A^{+} \b\|^2=O(\text{Rank}(A))=O(\min(m,n))$}. From our lower bounds analysis in the next section, these are indeed optimal.

\subsection{Lower bounds}

In this section, we show that the quantum protocols we are given in the previous section are optimal up to a factor of $\log (mn)$. We also prove the lower bounds of classical protocols for the task of sampling from the optimal solution.
To our ends, we first compute the quantum/classical communication complexity for the permutation-index problem defined as follows.

\begin{defn}[{\bf Permutation-Index Problem}]
Suppose Alice has a permutation $P=(P_1,\ldots,P_n)$ of $[n]$, \update{where $P_i\in [n]$}. Suppose Bob has an index $j\in[n]$. The goal is for Bob to determine $P_j$, where the communication is  1-way from Alice to Bob. 
\end{defn}

This problem is a special case of the index problem, in which $P$ is a multiset \cite{jayram2013optimal}. However, we shall show that it is as hard as the index problem.

\begin{prop}
The quantum and classical 1-way communication complexity of the {\bf Permutation-Index Problem} is $\Theta(n\log n)$.
\end{prop}

\begin{proof}
We prove that the index problem can be reduced to the {\bf Permutation-Index Problem}. Let $S=\{s_1,\ldots,s_n\}$ be a multiset of $n$ integers from $[n]$. In the index problem, Alice has $S$ and Bob has an index $j\in[n]$, their goal is for Bob to determine $s_j$. It is known that the quantum and classical 1-way communication complexity of the index problem is $\Theta(n\log n)$~\cite{jayram2013optimal,jain2008optimal}.

The reduction is as follows: In the first step, for any $i\in[n]$, Alice computes the multiplicity of $i$ in $S$, namely, Alice computes $m_i = \#\{j\in[n]:s_j=i\}$. Then she sends the information $ M = (m_1,\ldots,m_n)$ to Bob.  In the second step, Alice transforms the multiset $S$ into a permutation. Suppose $i_1<\cdots<i_p$ are the integers such that their multiplicities are nonzero. Then Alice replaces $i_1$ in $S$ with $1,2,\ldots,m_{i_1}$ (the order is not important in the replacement), replaces $i_2$ in $S$ with $m_{i_1}+1,\ldots,m_{i_1}+m_{i_2}$, and replaces $i_p$ in $S$ with $m_{i_1}+\cdots+m_{i_{p-1}}+1,\ldots,n$. In the end, Alice receives a permutation $P$. It is not hard to see that there is a one-to-one correspondence between $S$ and the pair $(M, P)$. So if there is a protocol for Alice and Bob to determine $P_j$, then from the above construction, they can use this protocol to determine $S_j$. Here, Alice needs to send $M$ to Bob first.

Next, we compute the communication complexity. In the first step, the number of bits required to transmit $M$ is bounded by $\log \binom{2n-1}{n-1}\leq 2n$, where $\binom{2n-1}{n-1} = \#\{m_1+\cdots+m_n=n:m_i\geq 0\}$ is the number of nonnegative $n$-decompositions of $n$. This means that the communication complexity of {\bf Permutation-Index Problem} is  $\Theta(n\log n) - 2n = \Theta(n\log n)$.
\end{proof}

We are now ready to prove the lower bounds of quantum and classical protocols. We first compute quantum lower bounds in terms of $\kappa$ and $\gamma$. Based on it, we then use similar ideas to prove the optimality of our quantum protocols and estimate classical lower bounds.  

\begin{thm}[Quantum lower bounds (with respect to $\kappa,\gamma$)]
\label{thm:Quantum lower bound}
Suppose Alice has a matrix $A \in \mathbb{R}^{m\times n}$ and Bob has a vector $\b \in \mathbb{R}^{m}$. To prepare the quantum state $\ket{A^{+}\b}$,
\begin{enumerate}
\item $\Omega(\kappa^2 + 1/\gamma^2 + \log \min(m,n))$ qubits of communication are required in the 1-way case from Bob to Alice.
\item $\Omega(\min(m,n)\log \min(m,n))$ qubits of communication are required in the 1-way case from Alice to Bob.
\item $\Omega(\kappa + 1/\gamma + \log \min(m,n))$ qubits of communication are required in the 2-way case.
\end{enumerate}
\end{thm}

\begin{proof}
To prove the claimed lower bounds,  our main idea is to reduce the disjointness problem or the index problem to a linear regression problem. In our reduction, the linear regressions we constructe have the property that $A$ is square. This naturally leads to lower bounds in the general case. Namely, if $m>n$ we can use the same reduction by focusing on 
$\argmin_{\x} \left\|\begin{pmatrix}
A \\
0
\end{pmatrix} \x - \begin{pmatrix}
\b \\
0
\end{pmatrix} \right\|$. If $m\leq n$, we can focus on 
$\argmin_{\x} \left\|\begin{pmatrix}
A & 0 \\
\end{pmatrix} \x - \begin{pmatrix}
\b & 0 \\
\end{pmatrix} \right\|$. This explains why the dependence on $m,n$ is $\min(m,n)$. Because of this, we below assume that $m=n$.

We prove the second claim using the hardness of {\bf Permutation-Index Problem}.  We can reduce the {\bf Permutation-Index Problem} to a linear regression problem as follows: 
Alice constructs a permutation matrix $P$ according to the permutation she has. Bob constructs the quantum state $\ket{j}$. If Bob can sample from the solution state $\ket{P_j} = P \ket{j}$, then Bob can solve the {\bf Permutation-Index Problem}. Thus the quantum lower bound of solving linear regression problems in case 2 is $\Omega(n\log n)$. 

We shall use the hardness of the disjointness problem to prove the first and third claims together. We aim to show that for any $\kappa$, there is an instance $(A,\b)$ such that \update{at least $\kappa$ (or $\kappa^2$) bits of communication are required} to prepare the quantum state of the optimal solution.

If $\kappa$ is too large, then the naive protocol of sending the whole matrix or vector will be used, so
we assume that $1\leq \kappa \leq n$. Denote $l=\lfloor \kappa \rfloor$ as the integer part of $\kappa$. 
Suppose Alice has a subset $S \subseteq [l]$ and Bob has another subset $T\subseteq [l]$. 
Without loss of generality, we assume that $S,T \neq \emptyset$ and $S, T$ are proper subsets of $[l]$ of size $\Theta(l)$.\footnote{\update{If $|T|=o(l)$, then Bob can send $T$ directly to Alice, which only uses $O(|T|)=o(l)$ bits of communication. So we assume this is not the case.}}
In the disjointness problem, their goal is to determine if $S\cap T=\emptyset$. It is known that for this problem, the quantum 1-way communication complexity is $\Theta(l)$ \cite{buhrman2001communication} and the quantum 2-way communication complexity is $\Theta(\sqrt{l})$ \cite{razborov2003quantum}, also see Proposition \ref{prop1}. 

The reduction is as follows: Choose $\varepsilon = 1/\sqrt{l}$. Alice constructs an $n\times n$ diagonal matrix $A$ as follows:
\[
A_{ii} = \begin{cases}
1 & i \in S, \\
1/\varepsilon & i \in [l] \backslash S, \\
0 & i \in  [n] \backslash [l].
\end{cases}
\]
Bob constructs an $n$ dimensional vector $\b$ as follows:
\[
b_i = \begin{cases}
1 & i \in T, \\
\varepsilon & i \in [l] \backslash T, \\
0 & i \in  [n] \backslash [l].
\end{cases}
\]
Then 
\[
\ket{A^+\b} = \frac{1}{\sqrt{L}} \left(\sum_{i\in S \cap T} \ket{i} + \varepsilon \sum_{j\in (S \backslash T) \cup (T \backslash S)} \ket{j} + \varepsilon^2 \sum_{k\in [m] \backslash {S \cup T} } \ket{k} \right) ,
\]
where
\[
L = |S \cap T| + \varepsilon^2 |(S \backslash T) \cup (T \backslash S)| + \varepsilon^4 |\overline{S \cup T}|.
\]

If $S \cap T \neq \emptyset$, then the probability of getting an $i\in S\cap T$ from measuring $\ket{A^{+}\b}$ is at least 1/2. If $S \cap T = \emptyset$, we obtain a uniformly random $i\in S \cup T$ from measurements. The disjointness problem is also hard even if $|S\cap T| \leq 1$ \cite{razborov1990distributional,razborov2003quantum}.\footnote{\update{This follows from the main theorem proved in~\cite{razborov2003quantum}. Indeed, in \cite{razborov2003quantum}, it was shown that given a predicate $D$ on $\{0,1,\ldots,n\}$, let $l_0:=\max\{l:1\leq l \leq n/2, D(l)\neq D(l-1)\}, l_1:=\max\{n-l:n/2\leq l <n, D(l)\neq D(l+1)\}$, then up to a logarithmic factor the bounded-error quantum communication complexity of $f(x,y):=D(|x\cap y|)$ is $\sqrt{n l_0}+l_1$. For the disjointness problem with the promise that $|S\cap T| \leq 1$, we have $D(l)=1$ if $l\in\{1,n\}$ and $D(l)=0$ otherwise. In this case, we have $l_0=1,l_1=1$. So the lower bound is $\sqrt{n}$. Here we added $D(n)=1$ intentionally to ensure $l_1$ is well-defined. It corresponds to the case that $S=T=[n]$, which is the trivial case.}} Under this setting if $S\cap T \neq \emptyset$ then we will see the same index from $S\cap T$ many times by measuring the state $\ket{A^{+}\b}$. If $S\cap T = \emptyset$, we will see different indices from $S\cup T$. {So preparing $\ket{A^{+}\b}$ is sufficient to solve the disjointness problem}.

It is easy to compute that  $\kappa = \sqrt{l}$ and $\gamma = \Theta(1)$. So the quantum communication complexity is at least quadratic in $\kappa$ in the first claim and at least linear in the third claim.

Regarding the dependence on $\gamma$, we use the following construction. We also assume that $|S\cap T|\leq 1$. If $|S\cap T|=1$, we denote the intersection as $\{w\}$. Alice constructs an $(n+1)\times (n+1)$ diagonal matrix $A$ by setting $A_{ii}=1$ if $i\in S \cup \{n+1\}$ and 0 otherwise. Bob constructs a vector $\b$ such that $b_i=1$ if $i\in T\cup \{n+1\}$ and 0 otherwise. Now we have $A^+ = A$ and $\x_{\rm opt} = A\b = \sum_{i\in S\cap T} \ket{i} + \ket{n+1}$. If $S\cap T = \emptyset$, we only see $n+1$ by measuring $\ket{\x_{\rm opt}}$. Otherwise, \update{we will see $w$ with probability $1/2$.
Now $\kappa = 1$ and $\gamma = \sqrt{1/|T|} = \Theta(\sqrt{1/l})$ since $|T|=\Theta(l)$ as assumed in the beginning.} Hence,  the quantum communication complexity is at least quadratic in $1/\gamma$ in the first claim and at least linear in the third claim.

Finally, the lower bound of $\log n$ comes from the index problem. In the index problem, Alice has a bit string $(x_1,\ldots,x_n)$ and Bob has an index $j$. The goal is to output $x_j$. If the communication is from Bob to Alice or 2-way, then Bob can just send the index to Alice, and Alice outputs $x_j$. This costs $\Theta(\log n)$ communication.\footnote{By \cite[Theorem 3.7]{kremer1999randomized}, it is known that the VC-dimension of this function is $\Theta(\log n)$. By \cite[Theorem 3]{klauck2000quantum-vc-dim}, VC-dimension is a lower bound of one-way quantum communication complexity. Thus the one-way quantum communication complexity is lower bounded by $\Omega(\log n)$.}
To build the connection between this problem and the linear regression problem, Alice constructs a permutation matrix $U\in \mathbb{R}^{2n \times 2n}$. It is a block-diagonal matrix, each block has dimension 2. If $x_j=1$, then the $j$-th block is $I_2$. Otherwise, the $j$-th block is Pauli-$X$. Bob constructs $\ket{0}\ket{j}$. So if $x_j=1$, then $U\ket{0}\ket{j} = \ket{0}\ket{j}$. Otherwise, $U\ket{0}\ket{j} = \ket{1}\ket{j}$. This means that if we can prepare $U\ket{0}\ket{j}$ we then can solve the index problem.
\end{proof}

Recall from Theorem \ref{thm:quantum protocol} that the communication complexity of the quantum protocol we proposed for case 1 is $O((\log m)\|A^{+}\|^2\|\b\|^2/\|A^{+} \b\|^2)$ and for case 3 is $O((\log m)\|A^{+}\|\|\b\|/\|A^{+} \b\|)$. For the constructions in the above proof, we have $\|A^+\| = 1$ and $\|A^{+} \b\|^2/\|\b\|^2  = \Theta(|S\cap T|/|T|)$. So we indeed proved that the lower bound is $\Omega(\|\b\|^2/\|A^{+} \b\|^2)$ for case 1 and $\Omega(\|\b\|/\|A^{+} \b\|)$ for case 3. \update{However, it is not clear what is the dependence on $\|A^+\|$. To understand this,}
we can make appropriate scaling so that the above construction shows that the lower bound is $\Omega(\|A^{+}\|^2\|\b\|^2/\|A^{+} \b\|^2)$ for case 1 and $\Omega(\|A^{+}\|\|\b\|^2/\|A^{+} \b\|)$ for case 3. This suggests that the quantum protocols for case 1 and case 3 are optimal up to a factor of $\log m$. We can use a similar construction to show the optimality of the quantum protocol for case 2 up to a factor of $\log(mn)$.
We state this in the following theorem. The proof is similar to that of Theorem \ref{thm:Quantum lower bound}, so we defer it to Appendix \ref{deferred proof}.

\begin{thm}[Quantum lower bounds]
\label{Quantum lower bounds-general}
Suppose Alice has a matrix $A \in \mathbb{R}^{m\times n}$ and Bob has a vector $\b \in \mathbb{R}^{m}$. To prepare the quantum state $\ket{A^{+}\b}$,
\begin{enumerate}
\item $\Omega(\|A^+\|^2\|\b\|^2/{\|A^+\b\|^2})$ qubits of communication are required in the 1-way case from Bob to Alice.
\item $\Omega(\|A^+\|_F^2\|\b\|^2/\|A^+\b\|^2)$ qubits of communication are required in the 1-way case from Alice to Bob.
\item $\Omega(\|A^+\|\|\b\|/\|A^+\b\|)$ qubits of communication are required in the 2-way case.
\end{enumerate}
\end{thm}

\update{In Theorem \ref{thm:Quantum lower bound}, the lower bound is additive with respect to $\kappa$ and $\gamma$, which are two quantities with nice explanations. Note that $\kappa = \|A\|\|A^+\|, \gamma = \|AA^+\b\|/\|\b\|$, so $\kappa/\gamma  \geq \|A^+\|\|\b\|/\|A^+\b\|$. Although the lower bound given in Theorem \ref{Quantum lower bounds-general} is multiplicative, we cannot say it is a stronger lower bound. We indeed did not prove that $\Omega(\kappa/\gamma)$ is a lower bound. Actually, theorem \ref{Quantum lower bounds-general} can be viewed as an alternative statement of Theorem \ref{thm:Quantum lower bound} using $A,\b$ rather than $\kappa, \gamma$.}


The classical communication complexity of the disjointness problem is $\Theta(n)$ in 2-way communication. So similar to the proof of Theorem \ref{thm:Quantum lower bound}, we have the following lower bounds for classical protocols.

\begin{thm}[Classical lower bounds]
\label{thm:Classical lower bound}
Suppose Alice has a matrix $A \in \mathbb{R}^{m\times n}$ and Bob has a vector $\b \in \mathbb{R}^{n}$. To sample from the solution $A^{+}\b$,
\begin{enumerate}
\item $\Omega(\min(m,n))$ bits communication are required in the 1-way case from Bob to Alice.
\item $\Omega(\min(m,n)\log \min(m,n))$ bits communication are required in the 1-way case from Alice to Bob.
\item $\Omega(\min(m,n))$ bits communication are required in the 2-way case.
\end{enumerate}
All the lower bounds are also true even if $A$ is well-conditioned, \update{i.e., $\kappa=O(1)$.}
\end{thm}

\begin{proof}
Similar to the analysis at the beginning of the proof of Theorem \ref{thm:Quantum lower bound}, we only need to consider the case that $m=n$.
The reductions in the proof of Theorem \ref{thm:Quantum lower bound} are also true for classical protocols, so we now only need to prove the claim that $\Omega(n)$ bits communication are required in the 2-way case even if $A$ is well-conditioned. Regarding this, we use the hardness of the Distributed Fourier Sampling problem studied in \cite{montanaro2019quantum}. \update{For convenience, we assume that $n=2^d$ for some integer $d>0$. In the Distributed Fourier Sampling problem, Alice has a function $f:\{0,1\}^d \rightarrow \{\pm 1\}$, Bob has another function $g:\{0,1\}^d \rightarrow \{\pm 1\}$.} Their goal is to sample from the distribution corresponding to the Fourier coefficients of $fg$, i.e., to sample from the state
\[
\ket{P_{fg}}:=\sum_{s\in\{0,1\}^n} \left( \frac{1}{2^d} \sum_{x\in\{0,1\}^d} f(x)g(x) (-1)^{s\cdot x}\right) \ket{s}.
\]
We can reduce this problem to a linear regression problem as follows. Alice constructs the matrix $H^{\otimes d} D_f$, where $H$ is the Hadamard matrix, and $D_f$ is diagonal with $x$-th diagonal entry equals $f(x)$, where $x\in\{0,1\}^d$. Bob constructs a vector whose quantum state is $\ket{g} = {2^{-d/2}} \sum_{x\in\{0,1\}^d} g(x) \ket{x}$. Then
$
\ket{P_{fg}} = H^{\otimes d} D_f \ket{g}.
$
It was shown in \cite[Theorem 1]{montanaro2019quantum} that any classical 2-way communication protocol for this problem must communicate $\Omega(2^d)=\Omega(n)$ bits, also see Proposition \ref{prop2}.
\end{proof}

Theorem \ref{thm:Classical lower bound} indicates that if the communication is from 1-way from Bob to Alice or 2-way, then the naive protocol for Bob sending the whole vector to Alice is optimal.
Regarding the second claim, we believe that the lower bound is $\Omega(n^2)$. However, we could not {prove this claim}. For this, we make the following conjecture.

\begin{con}
Suppose Alice has a matrix $A \in \mathbb{R}^{m\times n}$ and Bob has a vector $\b \in \mathbb{R}^{m}$. For Bob to sample from the solution $A^{+}\b$, $\Omega(mn)$ bits of communication are required in the 1-way case, where the communication is from Alice to Bob.
\end{con}

It is possible {that one may be able} to use the composition of the Distributed Fourier Sampling problem and the index problem to prove the conjecture. In this composed problem, Alice has $m$ Boolean functions $f_1,\ldots,f_m$, and Bob has a Boolean function $g$ as well as an index $j$. The goal is to sample from $\ket{P_{f_jg}}$.

\section{Multiple parties}

In this section, we consider linear regression problems in a more general setting. Suppose there are $s$ parties $P_0, \ldots, P_{r-1}$. For each $i$, the party $P_i$ receives a matrix $A_i \in \mathbb{R}^{d_i \times n}$ and a vector $\b_i\in \mathbb{R}^{d_i}$. Their goal is to solve the linear regression problem
\be
\label{general linear regression}
\argmin_{\x} \quad \|A \x - \b\|,
\ee
where
\be
\label{case 1}
A = \begin{pmatrix}
A_0 \\
\vdots \\
A_{r-1}
\end{pmatrix}, \quad 
\b = \begin{pmatrix}
\b_0 \\
\vdots \\
\b_{r-1}
\end{pmatrix}.
\ee
%


We assume that there is a referee such that each party can only send information to the referee. In the simultaneous message passing (SMP) model, the communication is 1-way, i.e., the referee is not allowed to send information to other parties. 
From the second claim of Theorem \ref{thm:Quantum lower bound}, we have that it is hard to solve the linear regression problem (\ref{general linear regression}) in the SMP model (see Proposition \ref{prop:lower bound SMP 1-way} below). 
So similar to the classical coordinator model \cite{vempala2020communication}, we assume that the communication is 2-way between each party and the referee. We call it the quantum coordinator model. We will {discuss this model in more detail} in Section \ref{section:Multiparty communication complexity of disjointness problem}.

\subsection{The quantum protocol}
\label{The quantum protocol:multiple parties}

In this section, we aim to propose a quantum protocol for solving (\ref{general linear regression}) based on the technique of quantum singular value transformation (QSVT). With QSVT, we have a near-optimal quantum algorithm for solving linear regression problems in terms of time and query complexity \cite{gilyen2019quantum}. Below, we show that this algorithm is still effective in the quantum coordinator model. \update{For completeness, we list all the invoked results in Appendix \ref{appC}.}

We first present a general result for QSVT in terms of communication complexity. 

\begin{prop}
\label{prop:QSVT in CC}
Suppose $P_i$ has a matrix $A_i \in \mathbb{R}^{d_i \times n}$, where $i\in\{0,\ldots,r-1\}$. Let $A$ be given in (\ref{case 1}). Assume that $A$ is Hermitian. Let $f\in \mathbb{R}[x]$ be a polynomial of degree $d$ satisfying that $|f(x)|\leq 1/2$ for all $x \in [-1,1]$.
Then there is a quantum protocol in the quantum coordinator model for the referee to construct an $(1, 3+\log r, 0)$ block-encoding of $f(A/\alpha)$ with $O(rd\log n)$ qubits of communication, where $\alpha = \sqrt{\sum_{i=0}^{r-1} \|A_i\|^2}$. Moreover, if $A=\sum_i A_i$, then the result still holds with $\alpha = \sum_i \|A_i\|$.
\end{prop}

\begin{proof} The quantum protocol contains two steps.

{\bf Step 1.} The referee needs a block-encoding of $A$. This is achieved by Lemma \ref{lem:block-encoding construction}.
The party $P_i$ constructs an $(\|A_i\|, 1, 0)$ block-encoding of $A_i$ based on SVD (see (\ref{block-encoding by SVD})), i.e., $P_i$ computes the following unitary
\[
U_i = \begin{pmatrix}
A_i/\|A_i\| & \cdot \\
\cdot & \cdot
\end{pmatrix}.
\]
By Lemma \ref{lem:block-encoding construction}, 
\be
\label{U}
U = ({\rm SWAP}_{1,2}\otimes I_n)\left(\sum_{i=0}^{r-1} \ket{i} \bra{i}  \otimes U_i\right) (V\otimes I_{2} \otimes I_n) 
\ee
is an $(\alpha,1+\log r,0)$ block-encoding of $A$, where $\alpha = \sqrt{\sum_i \|A_i\|^2}$, and 
\be \label{unitary V}
V \ket{0} = \frac{1}{\alpha} \sum_{j=0}^{r-1} \|A_j\| \, \ket{j}.
\ee
For the referee to use this unitary, each party $P_i$ sends $\|A_i\|$ to the referee so that the referee can construct the unitary $V$ satisfying
(\ref{unitary V}).
For any state, to apply $U$ to it, the referee can first apply $V\otimes I_{2} \otimes I_n$ to it. Next, the referee sends the state to $P_0$ and asks $P_0$ {to apply} $U_0$ to the second and third register if the first register is $\ket{0}$. After that, $P_0$ sends the state back to the referee so that the referee can ask $P_1$ to do a similar control operation based on $U_1$. They need to repeat this process $r$ times. Finally, the referee applies ${\rm SWAP}_{1,2}\otimes I_n$ to the resulting state. This process totally uses $O(r \log n)$ qubits of communication.

{\bf Step 2.} The referee constructs the block-encoding of $f(A/\alpha)$. This is achieved by QSVT. By \cite[Theorem 56]{gilyen2019quantum},  with the block-encoding of $A$ and the polynomial $f$, we can construct a $(1, 3+\log r,0)$ block-encoding $\widetilde{U}$ of $f(A/\alpha)$.
The interesting part is the quantum circuit of $\widetilde{U}$, which has a decomposition of the form (see \cite[Lemma 19]{gilyen2019quantum})
\be
\label{block-encoding of A}
\widetilde{U} = W_0 U W_1 U^\dag W_2 U W_3 U^\dag \cdots W_d U,
\ee
where $W_0, W_1,\ldots, W_d$ are generated by one- and two-qubit unitaries depending on the polynomial $f$ and some other public unitaries. Since these unitaries and function $f$ are public, the referee can use $W_0, W_1,\ldots, W_d$ without any communication. \update{As discussed in step 1, the referee can use $U$ and $U^\dag$ once with $O(r\log n)$ qubits of communication. In (\ref{block-encoding of A}), the referee uses $U$ and $U^\dag$ $O(d)$ times}. So to use $\widetilde{U}$ once they communicate $O( r d \log n)$ qubits in total.

The last claim can be proved similarly based on Lemma \ref{lem:block-encoding construction2}.
\end{proof}

Recall that in terms of time complexity, given an $(\alpha, q, \varepsilon)$ block-encoding of $A$ in cost $T$,  we can construct a $(1,q+2, 4d\sqrt{\varepsilon/\alpha})$ block-encoding  of $f(A/\alpha)$ in cost $O(dT)$ \cite[Theorem 56]{gilyen2019quantum}. By Proposition \ref{prop:QSVT in CC}, we still have the same formula for the communication complexity of using QSVT. The difference is that we can compute
$T=O(r \log n)$ and $\alpha=\sqrt{\sum_i \|A_i\|^2}$ precisely.
If $A$ is not Hermitian, the result in  Proposition \ref{prop:QSVT in CC} is also true except that the matrix function $f(A/\alpha)$ is defined with respect to singular value decomposition (see \cite[Definition 16]{gilyen2019quantum}). With the above proposition, we now can propose a quantum protocol for solving linear regressions. 


\begin{thm}
\label{thm:solve the general linear regression}
For the problem (\ref{general linear regression}) in the quantum coordinator model, there is a quantum protocol for the referee to prepare $\ket{A^{+}\b}$ by using 
$\widetilde{O}(r^{1.5}\kappa/\gamma)$
qubits of communication.
\end{thm}

\begin{proof}
\update{Let $\delta\in(0,1]$ be a threshold of the singular values of $A$. Our idea below depends on QSVT. When using QSVT to a polynomial approximation of $3\delta/4x$, we will obtain $A^+_{\geq \delta}$ automatically. Here $A^+_{\geq \delta}$ is the truncated matrix by removing the singular values of $A$ that are smaller than $\delta$.}
Just for the convenience of the statement of complexity analysis below, we assume that $\delta= \Theta(\sigma_{\min})$ so that $\|A\|/\delta = \Theta(\kappa)$ and $A^+_{\geq \delta}=A^+$.
Also for convenience, we denote $m = d_0+\cdots+d_{r-1}$, the row dimension of $A$. 

When solving linear regression problems, we can assume that $A$ is Hermitian. Otherwise we can consider
$
\begin{pmatrix}
0 & A \\
A^\dag & 0 
\end{pmatrix}.
$
Its block-encoding is 
$
\begin{pmatrix}
0 & U \\
U^\dag & 0 
\end{pmatrix}
$,
where $U$ is given in (\ref{U}).
With a similar argument, the referee can still use this block-encoding with $O(r \log (mn))$ qubits of communication. So below we assume that $A$ is Hermitian.

To apply Proposition \ref{prop:QSVT in CC}, the referee needs a polynomial approximation of $1/x$ in the interval $[-1,1]\backslash [-\delta',\delta']$. This function is public and its polynomial approximation is known, e.g., see \cite[Corollary 69]{gilyen2019quantum}. Indeed, \cite[Corollary 69]{gilyen2019quantum} gives a polynomial approximation of $3\delta'/4x$, which is enough for solving linear regression problems.
The degree of the polynomial is $d = O((1/\delta')\log(1/\varepsilon))$.
In Proposition \ref{prop:QSVT in CC}, $f$ will be applied to the singular values of $A$. However, the singular values of $A$ can be larger than 1. To overcome this, we can apply Proposition \ref{prop:QSVT in CC} to the matrix $A/\alpha$, where $\alpha = \sqrt{\sum_{i=0}^{r-1} \|A_i\|^2}=O(\sqrt{r}\|A\|)$ \update{because $\|A_i\|\leq \|A\|$ for all $i$}. This means $\delta'=\delta/\alpha$.
By Proposition \ref{prop:QSVT in CC}, there is a quantum protocol for the referee to construct a $(1,3+\log r,0)$ block-encoding 
$\widetilde{U}$ of $(3\delta/4)A^{+}$ with $O( (r\alpha/\delta) (\log 1/\varepsilon) \log (mn))$ qubits of communication in total. 

Regarding the quantum state of $\b$, each party sends the norm information of $\b_i$ to the referee, and then the referee prepares 
\be\label{state}
\frac{1}{\|\b\|} \sum_{i=0}^{r-1} \|\b_i\|  \, \ket{i} \ket{0}.
\ee
Similar to the application of the block-encoding of $A$, to prepare the target state 
\[
\ket{\b} = \frac{1}{\|\b\|} \sum_{i=0}^{r-1} \|\b_i\|  \, \ket{i} \ket{\b_i}
\]
the referee can send the state (\ref{state}) to each party gradually and ask that party to prepare $\ket{\b_i}$ using a control operator. This requires $O(r \log (mn))$ qubits of communication in total.

Finally, the referee applies $\widetilde{U}$ to $\ket{0} \ket{\b}$ to prepare 
\[
\frac{3\delta}{4} \ket{0} \otimes A^{+} \ket{\b} + \ket{0}^\bot.
\]
The success probability is $\Omega(\delta^2\gamma^2/\|A\|^2)$, where $\gamma$ is defined in (\ref{def:gamma}). Also, see a similar analysis in (\ref{lower bound of b}). Since the communication is 2-way, they can use amplitude amplification. Therefore, in total, they communicated
\[
O\left(\frac{r \alpha\|A\|}{\delta^2\gamma}
(\log 1/\varepsilon) \log(mn) \right)
=
O\left(r^{1.5} (\kappa^2/\gamma)
(\log 1/\varepsilon) \log(mn) \right)
\]
qubits.

The dependence on $\kappa$ can be reduced to be linear by the technique of variable-time amplitude amplification. This technique is still effective in the quantum coordinator model. We defer the analysis of this part to Appendix \ref{app}. 
\end{proof}

Recall that when solving linear regressions on a quantum computer, the time complexity is $\widetilde{O}((T_A+T_b)\alpha/\delta\gamma)$ \cite[Corollary 31]{chakraborty2019power}, where $T_A$ is the time complexity to construct the block-encoding of $A$ and $T_b$ is the time complexity to prepare the quantum state $\ket{\b}$. Using our notation in the communication complexity model, $\alpha = O(\sqrt{r}\|A\|)$ and $T_A =T_b = \widetilde{O}(r)$, where $T_A,T_b$ should be understood as the communication complexity of constructing the block-encoding and preparing the quantum state respectively. This leads to a complexity of $\widetilde{O}(r^{1.5}\kappa/\gamma)$, which is exactly the result described in Theorem \ref{thm:solve the general linear regression}. This means that the formula $\widetilde{O}((T_A+T_b)\alpha/\delta\gamma)$ is true for both time and communication complexity. The difference is that for communication complexity we can compute $T_A, T_b$ precisely, while for time complexity $T_A, T_b$ are usually hard to estimate.

\update{In the quantum case, the dependence of the complexity on $r$ is $r^{1.5}$, where $\sqrt{r}$ comes from the construction of the block-encoding of $A$ and $r$ comes from the number of parties. Regarding the time and query complexity, QSVT usually leads to the best algorithm for linear regression. So in the communication complexity, $r^{1.5}$ might be optimal. In comparison, the complexity is linear in $r$ classically \cite{vempala2020communication}. It was shown in \cite{vempala2020communication} that for the harder task of outputting a vector solution, the naive protocol, that is player $P_i$ sends $A_i^TA_i, A_i^T\b$ to the referee, is optimal. That's why the dependence on $r$ is linear. In the quantum case, some other techniques may be required if we aim to reduce the dependence on $r$.
}

Finally, we consider a general linear regression problem by setting
\be
\label{case 2}
A = A_0 + \cdots + A_{r-1}, \quad 
\b = \b_0 + \cdots + \b_{r-1}.
\ee
in (\ref{general linear regression}). Here we have to assume that $d_0=\cdots=d_{r-1}$. We also assume that $A,\b \neq 0$. Note that for $A,\b$ defined in (\ref{case 2}), the linear regression (\ref{general linear regression}) is equivalent to $\sum_i A_i^T A_i \x =\sum_i A_i^T \b_i$, which is a special case of the setting of (\ref{case 2}). For the setting (\ref{case 2}), by Proposition \ref{prop:QSVT in CC}, for any polynomial $f$ of degree $d$, the referee can construct a $(1,3+\log r,0)$ block-encoding of $f(A/\alpha)$ with $O(rd\log n)$ qubits of communication, where $\alpha = \sum_{i=0}^{r-1}\|A_i\|$. Regarding the quantum state of $\b$, the referee can prepare
\[
\frac{1}{\sum_{i} \|\b_i\|}  \sum_{i=0}^{r-1} \|\b_i\| \, \ket{0}\ket{\b_i} + \ket{0}^\bot
=
\frac{\|\b\|}{\sum_{i} \|\b_i\|}  \ket{0}\ket{\b} + \ket{0}^\bot
\]
by linear combination of unitaries with $O(s\log n)$ qubits of communication.\footnote{The proof is basically the same as Lemma \ref{lem:block-encoding construction2}.} Therefore, similar to the protocol in Theorem \ref{thm:solve the general linear regression}, we have the following result.

\begin{prop}
    For the setting (\ref{case 2}),
    there is a quantum protocol for the referee to prepare $\ket{A^{+}\b}$ by using
    \be
    \widetilde{O}\left(\frac{r \sum_{i=0}^{r-1}\|A_i\|}{\gamma \sigma_{\min}} \frac{\sum_{i=0}^{r-1} \|\b_i\|}{\|\b\|}  \right)
    \ee
    qubits of communication, where $\sigma_{\min}$ is the minimal nonzero singular value of $A$.
\end{prop}

Unlike Theorem \ref{thm:solve the general linear regression}, here we do not have $\sum_{i=0}^{r-1}\|A_i\| \leq r\|A\|$ and we also cannot give a nice bound for $\sum_{i=0}^{r-1} \|\b_i\|/\|\b\|$.


\subsection{Lower bounds}

In this section, we prove certain quantum/classical lower bounds for solving the linear regression problem (\ref{general linear regression}). First, we show that it is hard to solve the linear regression (\ref{general linear regression}) in the SMP model. This can be seen as evidence of why it is more interesting to consider the quantum coordinator model.

\begin{prop}
\label{prop:lower bound SMP 1-way}
Assume that $\sum_{i=0}^{r-1} d_i \geq n$.
In the SMP model, $\Omega(n\log n)$ qubits of communication are required to prepare the state $\ket{A^+\b}$, and $\Omega(n\log n)$ bits communication are required to sample from the solution $A^+\b$.
\end{prop}

\begin{proof}
This is a direct corollary of the second claim of Theorems \ref{thm:Quantum lower bound}, \ref{thm:Classical lower bound}. We can assume that the party $P_0$ knows $A_0,\ldots,A_{r-1}$ and the referee knows $\b_0,\ldots,\b_{r-1}$. Then this is equivalent to a linear regression problem in the Alice-Bob model, where the communication is 1-way from Alice to Bob.
\end{proof}




We can similarly prove lower bounds in the quantum/classical coordinator model. But this only gives lower bounds in terms of $\kappa$ or $n$.
Below, we consider the lower bound with respect to $s$ and provide a much stronger one. We will use the hardness of a multi-player set-disjointness problem considered in \cite{vempala2020communication}. In this problem, the party $P_j$ receives a subset $T_j \subseteq [n]$, and their goal is to determine if $T_0 \cap T_j \neq \emptyset$ for some $j\geq 1$. As shown in \cite[Theorem 3.1]{phillips2012lower} and \cite[Theorem 1]{woodruff2017distributed} that for any classical protocol that succeeds with probability $1-1/r^3$, the communication complexity is lower bounded by $\Omega(r n)$. In the quantum case, we have the following result.

\begin{lem}
\label{lem:quantum bound for s-OR}
    In the quantum coordinator model, the quantum communication complexity for the multi-player set-disjointness problem is $\Theta(\sqrt{n} r)$.
\end{lem}

\begin{proof}
Since the two-player set-disjointness problem can be solved with $O(\sqrt{n})$ qubits of communication, $O(r \sqrt{n} )$ provides a natural upper bound. Regarding the lower bound, we prove it in Section \ref{section:Multiparty communication complexity of disjointness problem}, see Theorem \ref{thm:multiparty disj}.
\end{proof}

\begin{thm}
\label{thm stronger lower bound}
Assume that $\sum_{i=0}^{r-1} d_i \geq n$, and $r = O(\sqrt{n})$.
In the coordinator model,  $\Omega(r \kappa)$ qubits of communication are required to prepare the state $\ket{A^+\b}$ and $\Omega(r n)$ bits communication are required to sample from $A^+\b$.
\end{thm}

\begin{proof}
The proof is based on the hardness of the multi-player set-disjointness problem discussed above. 
Let $\varepsilon = 1/\sqrt{n},  \xi = 1/\sqrt{r}$ and $\eta = 1/\sqrt{n}r$.
We consider the following reduction. The party $P_0$ constructs a diagonal matrix $D$ of dimension $n$ by setting the $i$-th diagonal entry as
\[
D_i = \begin{cases}
1 & i \in T_0, \\
1/\varepsilon & i \notin T_0.
\end{cases}
\]
For any $j\geq 1$, the party $P_j$ constructs a vector $\b_j$ by setting
the $i$-th entry as
\[
\b_j(i) = \begin{cases}
1 & i \in T_j, \\
\eta & i \notin T_j.
\end{cases}
\]
Using a similar idea to the proof of Theorem \ref{thm:Quantum lower bound}, we want to construct a linear regression problem such that the optimal solution is close to $D^{-1}(\b_1+\cdots+\b_{r-1})$. For this, we consider the following linear regression problem
\bes
\argmin_{\x} \|A\x-\b\|, \quad \text{ where }
A = \begin{pmatrix}
D \\
\xi I_n \\
\vdots \\
\xi I_n \\
\end{pmatrix}, \, 
\b = \begin{pmatrix}
\0 \\
\b_1 \\
\vdots \\
\b_{r-1} \\
\end{pmatrix}.
\ees
Up to normalization, the optimal solution is
\be
\label{solution}
\ket{\x_{\rm opt}} = (D^2 + (r-1)\xi^2 I_n)^{-1} \ket{\b_1+\cdots+\b_{r-1}}. 
\ee

It is easy to see that the $i$-th diagonal entry of $(D^2 + (r-1)\xi^2 I_n)^{-1}$ equals
\be
\label{entries}
\begin{cases}
\ds \frac{1}{1 + (r-1)\xi^2} & i \in T_0, \\
\ds \frac{1}{\varepsilon^{-2} + (r-1)\xi^2} & i \notin T_0.
\end{cases}
\ee
We use $c_i$ to denote the $i$-th entry of $\b_1+\cdots+\b_{r-1}$. Then it is easy to check that if $i\in T_1\cup\cdots\cup T_{r-1}$, we have $1\leq c_i\leq r-1$. Otherwise, $c_i=(r-1)\eta$. 

The quantum state of the optimal solution is 
\[
\ket{\x_{\rm opt}} =  \sum_{i\in T_0} \frac{c_i}{1 + (r-1)\xi^2} \, \ket{i} + \sum_{j\notin T_0} \frac{c_j}{\varepsilon^{-2} + (r-1)\xi^2} \, \ket{j}.
\]
We can reformulate it more precisely as follows
\beas
\ket{\x_{\rm opt}} &=&  \sum_{i\in T_0\cap (T_1\cup\cdots\cup T_{r-1})} \frac{c_i}{1 + (r-1)\xi^2}  \,  \ket{i} 
+\sum_{j\in T_0\backslash (T_1\cup\cdots\cup T_{r-1})} \frac{(r-1)\eta}{1 + (r-1)\xi^2}  \, \ket{j}  \\
&& +\, \sum_{k\in (T_1\cup\cdots\cup T_{r-1}) \backslash T_0} \frac{c_k}{\varepsilon^{-2} + (r-1)\xi^2}  \, \ket{k} 
+ \sum_{l\notin T_0\cup T_1\cup\cdots\cup T_{r-1} } \frac{(r-1)\eta}{\varepsilon^{-2} + (r-1)\xi^2}  \, \ket{l} ,
\eeas
where $1\leq c_i,c_k \leq r-1$.
The total probability weights before normalization of the last three summations are respectively bounded by
\[
n(r-1)^2\eta^2=O(1), \quad 
n(r-1)^2\varepsilon^4=O(1),\footnote{Here we used the fact that $r=O(\sqrt{n})$.} \quad 
n(r-1)^2\varepsilon^4\eta^2=O(\varepsilon^4). 
\]
The amplitude of the first summation is at least $\Omega(1)$ if $T_0\cap (T_1\cup\cdots\cup T_{r-1}) \neq \emptyset$. In this case, if measuring the state $\ket{\x_{\rm opt}}$ in the computational basis, we will see an index from the intersection with a probability of at least $1/3$.
We can assume that the size of the intersection has order 1 because the disjointness problem remains hard with this promise.
So if the intersection is nonempty, then we will see the same index many times. Otherwise, we will see many different indices uniformly.
This reduction shows that the lower bound for any classical protocol of solving linear regression (\ref{general linear regression}) is $\Omega(r n)$.

From (\ref{entries}), it is easy to see that the condition number $A$ is $\kappa = \Theta(1/\varepsilon) = \Theta(\sqrt{n})$. By Lemma \ref{lem:quantum bound for s-OR}, we obtain the claimed lower bound for quantum protocols.
\end{proof}


In \cite{chia2020sampling,jethwani2020quantum}, it was shown that QSVT can be dequantized, which implies that many quantum algorithms based on QSVT do not have exponential speedups in terms of time and query complexity. When studying communication complexity, we can still use QSVT due to Proposition \ref{prop:QSVT in CC}; however, the quantum speedups can be exponential in terms of communication complexity. This suggests that it is quite hard to use the techniques for dequantized algorithms to propose efficient classical protocols with low communication complexity.

\section{Hamiltonian simulation}

As a byproduct, in this section, we consider the problem of Hamiltonian simulation in the coordinator model. We define the problem as follows: Suppose $P_i$ holds a Hamiltonian $H_i$ of dimension $n$, the referee holds a quantum state $\ket{\psi}$, and their goal is to prepare the state $e^{i(H_0+\cdots+H_{r-1})t} \ket{\psi}$ quantumly or sample from it classically. By Proposition \ref{prop:QSVT in CC}, we can use QSVT to achieve the goal. The lower bounds analysis are also corollaries of the lower bounds we obtained previously.

We start from the simple case: the Alice-Bob model. Suppose Alice has a Hamiltonian $H$ of dimension $n$, Bob has a quantum $\ket{\psi}$, and their goal is to prepare the state $e^{iHt} \ket{\psi}$ quantumly or sample from it classically. As a corollary of Theorems \ref{thm:quantum protocol} and \ref{thm:Quantum lower bound}, we have the following result.

\begin{prop}
\label{prop:quantum-HS}
Suppose Alice has a Hamiltonian matrix $H \in \mathbb{C}^{n\times n}$ and Bob has a quantum state $\ket{\psi} \in \mathbb{C}^{n}$. Then the quantum communication complexity of outputting $e^{iHt} \ket{\psi}$ is
\begin{enumerate}
\item $\Theta(\log n)$ if the communication is 1-way from Bob to Alice or 2-way.
\item $\Theta(n \log n)$ if the communication is 1-way from Alice to Bob.
\end{enumerate}
\end{prop}

\begin{proof}
We apply Theorem \ref{thm:quantum protocol} to $A=e^{iHt}$ and $\b = \ket{\psi}$. Now $A$ is unitary.
\end{proof}

In the classical setting, the goal is to sample from the state $e^{iHt} \ket{\psi}$. Regarding the lower bound for classical protocols, we have the following result by Theorem \ref{thm:Classical lower bound}. The result is quite obvious because we can always write a unitary as $e^{iHt}$ for some $H$.

\begin{prop}
\label{prop:classical-HS}
Assume that Alice has a Hamiltonian matrix $H \in \mathbb{C}^{n\times n}$ and Bob has a vector $\ket{\psi} \in \mathbb{C}^{n}$. Then $\Omega(n)$ bits communication are required to sample from $e^{iHt} \ket{\psi}$ if the communication is 2-way.
\end{prop}

\begin{proof}
We still use the notation defined in the proof of Theorem \ref{thm:Classical lower bound}.
Note that the Hadamard matrix has the decomposition $H_2=e^{i\frac{\pi}{2}(I_2-H_2) }$. Let
\bes
L = \sum_{j=1}^d I_2^{\otimes (j-1)}\otimes (I_2-H_2)\otimes I_2^{\otimes (d-j)},
\ees
then $H_2^{\otimes d} = e^{i\frac{\pi}{2}L }$. In the proof of Theorem \ref{thm:Classical lower bound}, we can also consider the distribution $D_f H_2^{\otimes d} D_f \ket{g}$, which is equivalent to $H^{\otimes d} D_f \ket{g}$. Now we have $D_f H^{\otimes d} D_f = e^{i\frac{\pi}{2} D_f L D_f }$. So similar to the proof of Theorem \ref{thm:Classical lower bound}, 
Alice constructs the Hamiltonian $D_f L D_f$ and Bob constructs the quantum state $\ket{g}$. If they can sample from the resulting state, then they can solve the Distributed Fourier Sampling problem.
Hence, the lower bound of classical protocols is $\Omega(n)$.
\end{proof}


Finally, as an application of Proposition \ref{prop:QSVT in CC}, we consider the communication complexity of Hamiltonian simulation in the coordinator model when there are multiple parties.

\begin{prop}
\label{prop:SMP HS}
For any $i\in\{0,\ldots,r-1\}$, suppose the party $P_i$ receives a Hamiltonian $H_i$ of dimension $n$. Suppose the referee receives a quantum state $\ket{\psi}$. Then in the quantum coordinator model, there is a quantum protocol that costs 
\be
O\left((r\log n) \left(\sum_{i=0}^{r-1} \|H_i\| \, |t| + \frac{\log(1/\varepsilon)}{\log(e+(\sum_{i=0}^{r-1}  \|H_i\| \, |t|)^{-1}\log(1/\varepsilon))}  \right) \right)
\ee
qubits of communication to prepare the state $e^{i(H_0+\cdots+H_{r-1})t} \ket{\psi}$ up to error $\varepsilon$.
\end{prop}

\begin{proof}
By \cite[Lemma 59]{gilyen2019quantum}, {there is a polynomial that approximates $e^{it}$ up to error $\varepsilon$ with degree}
\[
d = O\left(|t| + \frac{\log(1/\varepsilon)}{\log(e+|t|^{-1}\log(1/\varepsilon))}  \right).
\]
By Proposition \ref{prop:QSVT in CC}, the referee can construct an $(1,3+\log r,0)$ block-encoding of $e^{it\sum_i H_i/\alpha}$ with $O(rd\log n)$ qubits of communication, where $\alpha = \sum_{i=0}^{r-1}  \|H_i\|$. We replace $t$ with $\alpha t$. Putting it all together, we obtain the claimed result.
\end{proof}

\section{Multiparty quantum communication complexity of disjointness}
\label{section:Multiparty communication complexity of disjointness problem}

In this section, we {complete the proof of our lower bounds in the coordinator model via proving bounds on} the quantum communication complexity of the disjointness problem in the multiparty case. We will consider a quantum model that is analogous to the classical coordinator model. Recall that in the coordinator model, there are $s$ parties $P_1,\ldots,P_r$, and there is a coordinator (here we call it the referee) $R$. The communication is 2-way between $P_i$ and $R$. If $P_i$ wants to send a message to $P_j$, then $P_i$ has to send the message to $R$ first, then $R$ will send the message to $P_j$. In the quantum case, we define a similar model. Different from the previous quantum multiparty model \cite{lee2009lower} which considers the blackboard model (i.e., if $P_i$ sends a message, then everyone else can see it), here we focus on the coordinator model (i.e., if $P_i$ sends a message, then only the referee can see the message). This model is almost equivalent to the message-passing {(``number in hand'')} model (i.e., no referee in this model, the party $P_i$ can send a message directly to another party $P_j$ and only $P_j$ can see the message) up to a factor of 2. 

In the model, we define the input as
\be
\ket{{\rm In}} = \ket{\phi(x_1)}_{P_1} \cdots \ket{\phi(x_r)}_{P_r} 
\ket{\vec{0}}_{C_1} \cdots \ket{\vec{0}}_{C_r} \ket{\phi(y)}_R , 
\ee
where $x_i$ is the initial information in $P_i$'s hand, $y$ is the initial information in the referee's hand. The states $\ket{\phi(x_i)}_{P_i}$, $\ket{\phi(y)}_{R}$ depend on the initial information. The register $\ket{\vec{0}}_{C_i}$ is the $i$-th channel. A quantum protocol is a quantum algorithm that applies a series of unitaries of forms
\be
U_{P_1,C_1}, \quad \cdots, \quad 
U_{P_r,C_r}, \quad U_{C_i,R}
\ee
to $\ket{{\rm In}}$. The unitary $U_{P_i,C}$ operates on the space of $P_i$ and the channel $C_i$. The unitary $U_{C_i,R}$ operates on the $i$-th channel and the space of the referee.
At the beginning of a quantum protocol, $P_1$ applies a unitary of the form $U_{P_1,C_1}$ to his space and the channel $C_1$. This corresponds to his private computation as well as to putting a message on the channel $C_1$. The length of this first message is the number of channel qubits affected by $P_1$'s operation. In the second round, the referee speaks and applies a unitary of the form $U_{C_1,R}$ to his space and the first channel. Then $P_2$ applies $U_{P_2,C_2}$, etc. If the referee speaks in the end, then a quantum protocol of $R:=2rt$ rounds defines an output state of the form
\be
\ket{{\rm Out}} = 
\prod_{i=1}^t
 \, 
U_{C_r,R}^{(i,2r)} \, U_{P_r,C_r}^{(i,2r-1)}
 \, \cdots  \, 
U_{C_1,R}^{(i,2)}  \,  U_{P_1,C_1}^{(i,1)} 
 \,  \ket{{\rm In}}.
\ee
\update{Here, for simplicity we assume that the number of rounds is a multiplier of $2r$.}
We assume that at the end of the protocol, the referee's register contains the answer. A measurement of this register then determines the output of the protocol. The quantum communication complexity is the number of qubits used in the whole procedure, which is $t(r+1)T$. Here $T$ is the total number of qubits in the channels.


We below consider the multiparty disjointness problem in this quantum coordinator model. The disjointness problem we are mainly interested in is defined as follows: $P_i$ has a subset $x_i$ of $[n]$, and the players aim to determine if there is an $i\geq 2$ such that $x_1 \cap x_i \neq \emptyset$. Equivalently, define the Boolean function that describes the 2-party disjointness problem as
\be \label{char fun of DISJ}
f(x,y) = \begin{cases}
    1 & |x\wedge y| \geq 1, \\
    0 & |x\wedge y| = 0.
\end{cases}
\ee
Then the disjointness problem defined above aims to compute
\be
f^r_{{\rm OR}}(x_1,x_2,\ldots,x_r) = f(x_1,x_2) \vee \cdots \vee  f(x_1,x_r) .
\ee

We use $Q_\varepsilon(f^r_{{\rm OR}})$ to denote the quantum communication complexity of computing $f^r_{{\rm OR}}$ with error $\varepsilon$. Namely, there is a quantum protocol without prior entanglement that computes $f^r_{{\rm OR}}$ of cost $Q_\varepsilon(f^r_{{\rm OR}})$ such that the acceptance probability on every $(x_1,x_2,\ldots,x_r)$ is at most $\varepsilon$ whenever $f^r_{{\rm OR}}(x_1,x_2,\ldots,x_r) = 0$ and at least $1-\varepsilon$ whenever $f^r_{{\rm OR}}(x_1,x_2,\ldots,x_r) = 1$. We use $Q^*_\varepsilon(f^r_{{\rm OR}})$ to denote the quantum communication complexity with prior entanglement.
The main result we aim to prove is as follows.

\begin{thm}
\label{thm:multiparty disj}
$Q^*_{\varepsilon}(f^r_{{\rm OR}}) = \Theta(r\sqrt{n})$.
\end{thm}

\begin{proof}
The upper bound is obvious. We below focus on the proof of the lower bound. For each $i\in\{2,\ldots,r\}$, let Bob plays the role of $P_i$ and Alice plays the role of the remaining parties as well as the referee. If there is a protocol that computes $f^r_{{\rm OR}}$, then the protocol allows us to determine if $x_1\cap x_i=\emptyset$ using at least $\Omega(\sqrt{n})$ qubits of communication by setting other subsets as the empty set. This means that in this protocol, $P_i$ needs to apply at least $\Omega(\sqrt{n})$ unitaries. Therefore, in total, the communication complexity is at least $\Omega(r\sqrt{n})$.
\end{proof}

\section{Connections between communication complexity and quantum-inspired classical algorithms}

\update{
In the quantum-inspired classical algorithms, we use a model that allows sampling and query (SQ) access to the input data. Using this model, it was proved that classically we could solve some problems, e.g., linear regressions, in cost polylog in the dimension in the low-rank case \cite{chia2020sampling}. We below discuss the connection between communication complexity and quantum-inspired classical algorithms. We will mainly focus on the Alice-Bob model.

First, we recall some definitions about quantum-inspired classical algorithms \cite[Definitions 2.5, and 2.10]{chia2020sampling}. For a vector $\b=(b_1,\ldots,b_n)\in\mathbb{C}^n$, we have $SQ(\b)$ if we can do the following three things: (i) for any $i$ we can query for $b_i$; (ii) we can sample from the distribution defined by $\text{Prob}(i)=|b_i|^2/\|\b\|^2$; (iii) we can query for the norm $\|\b\|$. For a matrix $A\in \mathbb{C}^{m\times n}$, we have $SQ(A)$ if (i) we have $SQ(A_{i*})$ for any $i$, where $A_{i*}$ is the $i$-th row of $A$; (ii) let $\a=(\|A_{1*}\|,\ldots,\|A_{m*}\|)$, then we  have $SQ(\a)$. 

For the linear regression problem $\argmin\|A\x-\b\|$, by a quantum-inspired classical algorithm of complexity $O(T)$ we mean we can compute $SQ(\x_*)$, where $\|\x_*-A^+\b\|\leq \varepsilon \|A^+\b\|$, by applying $SQ(A), SQ(\b)$ $O(T)$ times and $O(T)$ other arithmetic operations. For example, assuming $\b$ lies in the column spaces of $A$, then there is a quantum-inspired classical algorithm for linear regression with complexity $\widetilde{O}(\|A\|_F^4\|A\|^2\|A^+\|^6/\varepsilon^2)$ \cite{shao2021faster}. Without the assumption, the complexity is $\widetilde{O}(\|A\|_F^6\|A\|^6\|A^+\|^{12}/\varepsilon^4\gamma^2)$ \cite{gilyen2022improved}. 

In the Alice-Bob model, we assume the communication is 2-way. By communicating with each other once, Alice can use $SQ(\b)$ or Bob can use $SQ(A)$ once. Therefore, it is easy to obtain the following result.

\begin{prop}
If there is a quantum-inspired classical algorithm for $\argmin\|A\x-\b\|$ of complexity $O(T)$, then there is a classical protocol to solve this linear regression in the Alice-Bob model of communication complexity $O(T)$, where Alice holds $A$ and Bob holds $\b$, the communication is 2-way, and the goal is to sample from a distribution $\varepsilon$-close to the one defined by $\ket{A^+\b}$.
\end{prop}

Similar to the proof of Theorem \ref{thm:Classical lower bound}, using the hardness of the Distributed Fourier Sampling problem (see Proposition \ref{prop2}), it is easy to conclude that in the low-rank case, the classical communication complexity is lower bounded by the Rank($A$), while the quantum communication is $O(1)$ for well-conditioned linear regressions.

Next, let us see two examples that suggest that low rank is not the only assumption for the efficiency of quantum-inspired classical algorithms. We consider the disjointness problem. Recall that in this problem, Alice and Bob respectively have $\a=(a_1,\ldots,a_n), \b=(b_1,\ldots,b_n)\in \{0,1\}^n$, they want to determine if there is an $i$ such that $a_i=b_i=1$. Without loss of generality, we assume that the hamming weights $|\a|,|\b|=\Theta(n)$. Consider the following construction
\[
A=\ket{0}\bra{0} + \frac{1}{n} \ket{\a}\bra{\a},
\quad
\b =  \ket{0} + \ket{\b},
\]
where $\ket{\a}, \ket{\b}$ are the quantum states of $\a,\b$ respectively. Now $A$ has rank 2, and the solution is
\[
A^+\b =  \ket{0} + n \braket{\a|\b} \ket{\a}.
\]
If there is no $i$ such that $a_i=b_i=1$, then $\braket{\a|\b}=0$, so $A^+\b = \ket{0}$. If there is an $i$ such that $a_i=b_i=1$, then we have $A^+\b \approx \ket{0} + \ket{\a}$. Here we assumed that there is only one such $i$, which is the worst case. Thus, if we can sample from the solution, then in the latter case, we will see some indices from $\{1,\ldots,n\}$ with probability 1/2. As a result, we can solve the disjointness problem. This means $\Omega(n)$ bits of communication are required to solve this linear regression. It also means that to solve this linear regression, any quantum-inspired classical algorithm costs $\Omega(n)$. In this example, $A$ has a low rank, while the complexity is linear in $n$. This is indeed not a contradiction. In this example, we have $\|A\|_F = \Theta(1), \|A^+\|= n$, and usually quantum-inspired classical algorithms are highly affected by $\|A\|_F\|A^+\|$, which is $\Theta(n)$ now. 

Let us below consider another example, we set
\[
A=\ket{0}\bra{0} + \ket{\a}\bra{\a},
\quad
\b = \frac{1}{n}  \ket{0} + \ket{\b}.
\]
Now the solution is
\[
A^+\b = \frac{1}{n} \ket{0} +  \braket{\a|\b} \ket{\a}.
\]
If there is no $i$ such that $x_i=y_i=1$, then $A^+\b = \frac{1}{n} \ket{0}$. Otherwise, we have $A^+\b \approx \frac{1}{n}(\ket{0} + \ket{\a})$. Similarly, by measuring, we can also solve the disjointness problem. In this example, $\|A\|_F\|A^+\|=\sqrt{2}$. However, now $\b$ is far away from the column space of $A$, i.e., $\gamma:=\|AA^+\b\|/\|\b\| \approx 1/n $. So similar to quantum algorithms \cite{chakraborty2019power}, quantum-inspired classical algorithms for linear regressions are also affected by $\gamma$.

Usually, it is not easy to analyze the lower bounds for classical computation, and communication complexity provides us with an efficient tool to prove some nontrivial lower bounds. So it is possible that we can find some other interesting properties of quantum-inspired classical algorithms through communication complexity.

}

\section{Conclusions}

In this work, we showed that quantum computers have provable polynomial or exponential speedups for solving linear regression problems and Hamiltonian simulation in terms of communication complexity. We also found that in the quantum coordinator model, we can still efficiently use the quantum singular value transformation technique. Because of this, we believe that for many other linear algebra problems, \update{it is possible to obtain provable quantum speedups using this technique in terms of communication complexity.}

\subsection*{Acknowledgements}

We acknowledge support from EPSRC grant EP/T001062/1. This project has received funding from the European Research Council (ERC) under the European Union's Horizon 2020 research and innovation programme (grant agreement No.\ 817581). No new data were created during this study.

\appendix

\section{Proof of Theorem \ref{Quantum lower bounds-general}}

\label{deferred proof}

From the proof of Theorem \ref{thm:Quantum lower bound}, without loss of generality, we can simply assume that $S,T \subseteq [n]$ and $|S|,|T|=\Theta(n)$.

We prove the first and third claims together.
Alice and Bob respectively construct a diagonal matrix $A$ and a vector $\b$ by setting
\[
A_{ii} = \begin{cases}
\sqrt{\varepsilon} & i \in S, \\
1/\sqrt{\varepsilon} & i \notin S.
\end{cases}
\quad 
b_i = \begin{cases}
1/\sqrt{\varepsilon} & i \in T, \\
\sqrt{\varepsilon} & i \notin T,
\end{cases}
\]
where $\varepsilon = 1/\sqrt{n}$.
Then the optimal solution is
\[
\x_{\rm opt} = \frac{1}{\varepsilon} \sum_{i\in S \cap T} \ket{i} + \sum_{j\in (S \backslash T) \cup (T \backslash S)} \ket{j} + \varepsilon \sum_{k\in \overline{S \cup T} } \ket{k}.
\]
So
\beas
\|\x_{\rm opt}\|^2 = \frac{|S \cap T|}{\varepsilon^2} 
+ |(S \backslash T) \cup (T \backslash S)|
+ \varepsilon^2 |\overline{S \cup T}|.
\eeas

If $S\cap T \neq \emptyset$, then the norm of $\x_{\rm opt}$ is dominated by the first term, so $\|\x_{\rm opt}\|^2 = n |S\cap T|$.
We can also compute that $\|A^+\|^2 = 1/\varepsilon = \sqrt{n}$ and
\[
\|\b\|^2 = \frac{|S \cap T|}{\varepsilon} + \varepsilon |S \backslash T| +  \frac{|T \backslash S|}{\varepsilon} + \varepsilon |\overline{S \cup T}| = \Theta(\sqrt{n} |T|).
\]
Therefore,
\[
\frac{\|A^+\|^2 \|\b\|^2}{\|\x_{\rm opt}\|^2}
= \Theta\left(\frac{|T|}{|S\cap T|}\right).
\]

If $S\cap T = \emptyset$, then $\|\x_{\rm opt}\|^2 = \Theta(|S|+|T|) = \Theta(n)$. So 
\[
\frac{\|A^+\|^2 \|\b\|^2}{\|\x_{\rm opt}\|^2}
= \Theta(n) =\Theta(|T|).
\]
Note that $\Theta({|T|}/{\max(1,|S\cap T|)})$ is the quantum communication complexity for the disjointness problem if the communication is 1-way. If it is 2-way, the complexity is $\Theta(\sqrt{{|T|}/{\max(1,|S\cap T|)}})$.

Below, we prove the second claim. We will use the hardness of the index problem. Alice has a (0,1)-matrix $A\in \mathbb{R}^{m\times n}$ and Bob has an index $(i,j)$, their goal is to determine $A_{ij}$. It is known that the communication complexity of this problem is $\Theta(mn)$. Without loss of generality, we assume that the number of 1s in each column of $A$ is $\Theta(m)$, and the number of 1s in each row is $\Theta(n)$.\footnote{If some columns of $A$ contain $o(m)$ 1's, then Alice can send all these columns to Bob first. This totally costs $o(mn)$, which is strictly less than $mn$ 
So removing these columns does not affect the hardness of the index problem. The same analysis is also true for rows. Hence, we can assume that each column has $\Theta(m)$  1's and each row has $\Theta(n)$ 1's.} We can reduce the index problem to a linear regression problem using a similar construction to the above. For the $j$-th column, Alice constructs a diagonal matrix $D_j$ as follows: We use $D_j(k,k)$ to denote the $k$-th diagonal entry, then define
\[
D_j(k,k) = \begin{cases}
1/\sqrt{\varepsilon} & A_{kj}=1, \\
\sqrt{\varepsilon} & A_{kj}=0.
\end{cases}
\]
Now $\varepsilon=1/\sqrt{m}$.
With the index $(i,j)$, Bob constructs a vector $\b_j$ as follows: We use $\b_j(k)$ to denote the $k$-th entry, then define
\[
\b_j(k) = \begin{cases}
1/\sqrt{\varepsilon} & k=i, \\
\sqrt{\varepsilon} & k\neq i.
\end{cases}
\]
In the end, they consider the linear regression problem
$
\min_{\x} \|D\x-\b\|,
$
where
\[
D = \begin{pmatrix}
    D_1 \\
    \vdots \\
    D_n
\end{pmatrix}_{mn\times m}, \quad 
\b = \begin{pmatrix}
    \0_{m(j-1)} \\
    \b_j \\
    \0_{m(n-j)} \\
\end{pmatrix}_{mn\times 1}.
\]
Here $\0_{m(j-1)}, \0_{m(n-j)}$ are zero vectors of length $m(j-1)$ and $m(n-j)$ respectively. 

Note that the pseudoinverse of $D$ is
\[
D^+  = \left(\sum_{i=1}^n D_i^2\right)^{-1}
\begin{pmatrix}
    D_1 & \cdots & D_n
\end{pmatrix}.
\]
Thus the optimal solution of the above constructed linear regression problem is
\[
\x_{\rm opt} = \left(\sum_{i=1}^n D_i^2\right)^{-1} D_j \b_j.
\]
For convenience, we denote the $k$-th entry of $\sum_{i=1}^n D_i^2$ as $d_k$. Then
\be \label{diagonal entry}
d_k = \sum_{i=1}^n D_i(k,k)^2
= 
\frac{1}{\varepsilon} \sum_{i:A_{ki}=1} 1
+\varepsilon \sum_{i:A_{ki}=0} 1
 = \Theta(n/\varepsilon) = \Theta(n\sqrt{m}),
\ee
where we used the assumption that each column of $A$ has $\Theta(n)$ 1's.
If $A_{ij}=1$, then we can reformulate $\x_{\rm opt}$ as follows:
\[
\x_{\rm opt} = 
\frac{1}{\varepsilon d_i} \ket{i} + \sum_{k:k\neq i, A_{kj}=1} \frac{1}{d_k} \ket{k} + \varepsilon \sum_{k:A_{kj}=0} \frac{1}{d_k} \ket{k}.
\]
By measuring this state, we will see $i$ with a constant probability. So we will see $i$ many times when repeating the measurements.
If $A_{ij}=0$, then the first term does not exist and the second term is summing over all $k$ with $A_{kj}=1$. In this case, we will see many different indices by measuring the solution state.

We now estimate the communication complexity of our quantum protocol for this linear regression problem. First, we can compute that
\beas
\|\x_{\rm opt}\|^2 &=&  \frac{1}{\varepsilon^2 d_i^2} + \sum_{k:k\neq i, A_{kj}=1} \frac{1}{d_k^2} + \varepsilon^2 \sum_{k:A_{kj}=0} \frac{1}{d_k^2} = \Theta(1/\varepsilon^2n^2m) = \Theta(1/n^2), \\
\|\b\|^2 &=& \frac{1}{\varepsilon} + \varepsilon (m-1) = \Theta(\sqrt{m}).
\eeas
Regarding the Frobenius norm of $D^+$, note that we assumed that each column of $A$ has $\Theta(m)$ 1's, so we have
\[
\|D^+\|_F^2 =
\sum_{j=1}^n \left\|\frac{D_j}{\sum_{i=1}^n D_i^2}\right\|_F^2
= \Theta\left(n\left\|\frac{D_1}{\sum_{i=1}^n D_i^2}\right\|_F^2\right)
= \Theta\left( \frac{n\|D_1\|_F^2}{n^2m}\right)
= \Theta\left( \frac{nm/\varepsilon}{n^2m}\right)
= \Theta\left( \frac{\sqrt{m}}{n}\right).
\]
In the above, the second equality is caused by the facts that $d_k = \Theta(n\sqrt{m})$ from (\ref{diagonal entry}) 
and that each column of $A$ has $\Theta(m)$ 1's so that $\|D_j\|_F^2 = \Theta(\|D_1\|_F^2)$.
Thus
\[
\frac{\|D^+\|_F^2\|\b\|^2}{\|\x_{\rm opt}\|^2}
= \Theta(mn).
\]
This matches the complexity of the index problem.

\section{Further details of the proof of Theorem \ref{thm:solve the general linear regression}}

\label{app}

\update{
In this appendix, we briefly describe the variable-time quantum algorithm (VTAA) for preparing $\ket{A^{-1}\b}$ and show that it still works in the quantum coordinator model. The following definition comes from \cite[section 5]{childs2017quantum}, which originally from \cite[section 3.3]{ambainis2012variable}.

\begin{defn}
Let $\mathcal{A}$ be a quantum algorithm on a space $\mathcal{H}$ that starts in the state $\ket{0}_\mathcal{H}$. We say $\mathcal{A}$ is a variable-time quantum algorithm if the following conditions hold:

\begin{enumerate}
    \item $\mathcal{A}$ can be written as the product of $T$ algorithms $\mathcal{A}=\mathcal{A}_T \cdots \mathcal{A}_2\mathcal{A}_1$.
    \item  $\mathcal{H}$ can be written as a product  $\mathcal{H}=\mathcal{H}_C\otimes \mathcal{H}_A$, where $\mathcal{H}_C$ is a product of $T$ single qubit registers denoted by $\mathcal{H}_{C_1},\ldots,\mathcal{H}_{C_T}$.
    \item Each $\mathcal{A}_j$ is a controlled unitary that acts on the registers $\mathcal{H}_{C_j} \otimes \mathcal{H}_A$ controlled on the first $j-1$ qubits of $\mathcal{H}_{C}$ being set to 0.
\end{enumerate}

\end{defn}

}

In VTAA, two key techniques are performing gapped quantum phase estimation on $A$ and computing truncated block-encoding of $A^{-1}$. We first state these two results and then check why they are still working in the quantum coordinator model. The following result comes from \cite[Lemma 22]{childs2017quantum}.

\begin{lem}[Gapped Phase Estimation (GPE)]
\label{appB:lem1}
Let $U$ be a unitary such that $U \ket{\psi} = e^{i\lambda} \ket{\psi}$ and $\lambda \in[-1,1]$. Let $\phi \in (0,1/4]$ and $\varepsilon>0$. Then there is a quantum algorithm that maps
\[
\ket{0} \ket{0} \ket{\psi} \mapsto 
\alpha_0 \ket{0} \ket{g_0} \ket{\psi}
+\alpha_1 \ket{1} \ket{g_1} \ket{\psi}
\]
for some unit vectors $\ket{g_0}, \ket{g_1}$, and 
\begin{itemize}
    \item if $0 \leq |\lambda| \leq \phi$, then $|\alpha_1| \leq \varepsilon$,
    \item if $2\phi \leq |\lambda| \leq 1$, then $|\alpha_0| \leq \varepsilon$.
\end{itemize}
If $T_U$ is the cost of implementing $U$, then the cost of this quantum algorithm is $O((\log1/\varepsilon)T_U/\phi)$.
\end{lem}

The quantum algorithm in the above lemma is based on the standard phase estimation. To use phase estimation in the quantum coordinator model, the main obstacle for them is the Hamiltonian simulation. In our case, $U=e^{iA}$, where $A$ is given in (\ref{case 1}). So the referee needs to carry out Hamiltonian simulation of $A$. This is achieved by QSVT. A quantum protocol can be given in a similar way to that of Proposition \ref{prop:SMP HS}. The communication complexity of using $U$ is $T_U=\widetilde{O}(r\alpha)$, where $\alpha = \sqrt{\sum_i \|A_i\|^2}$. Note that if we estimate the time complexity, then $T_U=\widetilde{O}(T_A \tilde{\alpha})$, where $T_A$ is the cost to construct a block-encoding of $A$ and $\tilde{\alpha} \geq \alpha$ generally. So for communication complexity we can say $T_A=\widetilde{O}(r)$.


Another result that will be used in the VTAA is truncated block-encoding of $A^{-1}$, see \cite[Corollary 29]{chakraborty2019power}. 

\begin{lem}
\label{appB:lem2}
Let $A$ be Hermitian, and let $U$ be an $(\alpha, a, \varepsilon)$ block-encoding of $A$ that can be implemented using $T_A$ elementary unitaries. Then for any state $\ket{\psi}$ that is spanned by eigenvectors of $A$ with eigenvalues in the range $[-1,-\lambda]\cup[\lambda,1]$, there exists a unitary $W(\lambda,\varepsilon)$
\[
W(\lambda,\varepsilon): \ket{0} \ket{0} \ket{\psi}
\mapsto
\frac{1}{\alpha_{\max}} \ket{1} \ket{0} f(A) \ket{\psi}
+ \ket{0}^\bot,
\]
where $\alpha_{\max} \leq \lambda$ is a constant and $\|f(A)\ket{\psi} - A^{-1}\ket{\psi}\| \leq \varepsilon$. The cost of implementing $W(\lambda,\varepsilon)$ is
$\widetilde{O}((a+T_A)\alpha/\lambda)$.
\end{lem}

For us, $A$ is given in (\ref{case 1}). The unitary $W(\lambda,\varepsilon)$ is obtained in a similar way to the block-encoding of $(3\delta/4) A^+$ defined in (\ref{block-encoding of A}), where we focused on singular values that are at least $\delta$. Now we need to focus on singular values that are at least $\lambda$.  In the communication complexity, $a=\log r$ and $T_A=O(r\log(n))$ is the required number of qubits of communication to run $U$. So for the referee to use $W(\lambda,\varepsilon)$, they need to communicate $\widetilde{O}(r\alpha/\lambda)$ qubits in total, where $\alpha = \sqrt{\sum_i \|A_i\|^2}$.

\update{We next briefly describe the variable-time quantum algorithm $\mathcal{A}$. For more, especially about the correctness and complexity analysis, we refer to \cite[section 5]{childs2017quantum}. The algorithm $\mathcal{A}$ is built as a sequence of steps $\mathcal{A}_1,\ldots,\mathcal{A}_T$ with $T = \lceil \log \kappa \rceil +1$, so the algorithm is $\mathcal{A} = \mathcal{A}_T\ldots,\mathcal{A}_1$. The algorithm $\mathcal{A}$ uses the following registers:

\begin{itemize}
\item a $T$-qubit clock register $C$, labelled $C_1,\ldots, C_T$, used to determine a region the eigenvalue belongs to (i.e., to store the result of GPE);
\item a single-qubit flag register $F$ to indicate whether the approximation of $A^{-1}$ was successfully implemented;
\item a $(\log n)$-qubit register $I$, initialised to $\ket{\b}$, that finally contains the output state;
\item a register $P$, divided into registers $P_1, \ldots, P_T$, to be used as ancilla for GPE;
\item a register $Q$ to be used as ancilla in the implementation of $A^{-1}$.
\end{itemize}

The corresponding Hilbert spaces are denoted by $\mathcal{H}_C ,\mathcal{H}_F , \mathcal{H}_I , \mathcal{H}_P$, and $\mathcal{H}_Q$, respectively. All registers are initialized in $\ket{0}$, except for register $I$, which is initialized in $\ket{\b}$. When we write $\ket{0}_X$ we mean that all qubits of register $X$ are in $\ket{0}$.

We now describe algorithm $\mathcal{A}_j$. In the algorithm below, each call to GPE uses the unitary operator $e^{iA}$. For all $j\in [T]$, let $\varphi_j = 2^{-j}$, and let $\delta = \varepsilon/(T \alpha_{\max})$. We define $\mathcal{A}_j$ as the product of the following two unitary operations:

\begin{enumerate}
\item Conditional on the first $(j-1)$ qubits of $\mathcal{H}_C$ being $\ket{0}$, apply GPE($\varphi_j,\delta$) on the input state in $I$ using $C_j$ as the output qubit and additional fresh qubits from $P$ as ancilla (denoted by $P_j$).
\item Conditional on $C_j$ (the outcome of the previous step) being $\ket{1}_{C_j}$, apply $W(\varphi_j,T\delta)$ to the input state in $I$ using $F$ as the flag register and register $Q$ as ancilla.
\end{enumerate}
}

As we can see, Lemmas \ref{appB:lem1} and \ref{appB:lem2} are the two main tools in VTAA. We have checked that they are still working in the quantum coordinator model, and so is VTAA.  Compared with the time complexity, we can see that the only difference in the communication complexity is that $T_A$ can be computed precisely in the above two lemmas. Roughly, for time complexity, it may not be easy to compute $T_A$, while for communication complexity, $T_A=\widetilde{O}(r)$. Also, the parameter $\alpha$ in the block-encoding can be computed precisely too for communication complexity. In summary, the time complexity results for the above two lemmas still hold for communication complexity, while for communication complexity, we know the values of $T_A$ and $\alpha$.

\section{Some previous results about quantum singular value transformation}
\label{appC}

\update{

In this appendix, we collate the results that we will use about quantum singular value transformation.

\begin{defn}[Alternating phase modulation sequence, definition 15 of \cite{gilyen2019quantum}]
\label{defn1:QSVT}
Let $\mathcal{H}_U$ be a finite dimensional Hilbert space, and let $U, \Pi, \widetilde{\Pi} \in \text{End}(\mathcal{H}_U)$ be linear operators on $\mathcal{H}_U$ such that $U$ is unitary and $ \Pi, \widetilde{\Pi}$ are orthogonal projectors. Let $\Phi \in \mathbb{R}^n$, then we define the phased alternating sequence $U_\Phi$ as follows:
\[
U_\Phi
=
\begin{cases} \vspace{.1cm}
e^{i\phi_1}(2\widetilde{\Pi}-I) U \prod_{j=1}^{(n-1)/2} \Big(e^{i\phi_{2j}(2\Pi-I)} U^\dag e^{i\phi_{2j}(2\widetilde{\Pi}-I)} U \Big), & n \text{ is odd,} \\
\prod_{j=1}^{n/2} \Big(e^{i\phi_{2j}(2\Pi-I)} U^\dag e^{i\phi_{2j}(2\widetilde{\Pi}-I)} U \Big), & n \text{ is even.}
\end{cases}
\]

\end{defn}

\begin{lem}[Efficient implementation of alternating phase modulation sequences, Lemma 19 of \cite{gilyen2019quantum}]
The alternating phased sequence $U_\Phi$ can be implemented using a single ancilla qubit with $n$ uses of $U$ and $U^\dag$, $n$ uses of $C_{\Pi}NOT$ and $n$ uses of $C_{\widetilde{\Pi}}NOT$ gates and
$n$ single qubit gates.
\end{lem}

\begin{defn}
Let $f:\mathbb{R} \rightarrow \mathbb{C}$ be an even or odd function, let $A\in \mathbb{C}^{m\times n}$ and $A = \sum_{i=1}^{\min(m,n)} \sigma_i \ket{u_i} \bra{v_i}$ be the singular value decomposition of $A$. If $f$ is odd, we define
$f^{\text{(SV)}}(A) = \sum_{i=1}^{\min(m,n)} f(\sigma_i) \ket{u_i} \bra{v_i}$. If $f$ is odd, we define $f^{\text{(SV)}}(A) = \sum_{i=1}^{n} f(\sigma_i) \ket{v_i} \bra{v_i}$, where $\sigma_i = 0$ for $i\in [n]\backslash[\min(m,n)]$.
\end{defn}

\begin{thm}[Corollary 8 and Theorem 17 of \cite{gilyen2019quantum}]
Using the same notation as in Definition \ref{defn1:QSVT}, for any even or odd polynomial $f(x)$ of degree $n$ of the following properties:

\begin{itemize}
    \item $\forall x\in[-1,1]: |f(x)|\leq 1$,
    \item $\forall x\in(-\infty,-1]\cup [1,\infty): |f(x)| \geq 1,$
    \item if $n$ is even, then $f(ix)\bar{f}(ix)\geq 1$ for all $x\in \mathbb{R}.$
\end{itemize}
There is an efficiently computable $\Phi \in \mathbb{R}^n$ such that
\[
f^{\text{(SV)}}(\widetilde{\Pi} U \Pi) = 
\begin{cases}
    \widetilde{\Pi} U_{\Phi} \Pi & \text{if } n \text{ is odd,} \\
    \Pi U_{\Phi} \Pi & \text{if } n \text{ is even.}
\end{cases}
\]
\end{thm}

\begin{prop}[Corollary 69 of \cite{gilyen2019quantum}]
Let $\varepsilon, \delta \in (0, 1/2 ]$, then there is an odd polynomial $P \in \mathbb{R}[x]$ of degree $O(\delta^{-1} \log(1/\varepsilon))$ that is $\varepsilon$ approximating $f(x) = 3\delta/4x$ on the domain $[-1,1]\backslash[-\delta,\delta]$, moreover, it is bounded 1 in absolute value.
\end{prop}

\begin{prop}[Lemmas 57 and 59 of \cite{gilyen2019quantum}]
Let $t\in \mathbb{R}\backslash\{0\}, \varepsilon\in(0,1/e),$ and let $R=\lfloor r( \frac{e|t|}{2}, \frac{5\varepsilon}{4})/2 \rfloor$, then
\beas
&& \|\cos(tx) - J_0(t) + 2 \sum_{k=1}^R (-1)^k J_{2k}(t) T_{2k}(x) \|_{[-1,1]} \leq \varepsilon, \\
&& \|\sin(tx) - 2 \sum_{k=0}^R (-1)^k J_{2k+1}(t) T_{2k+1}(x) \|_{[-1,1]} \leq \varepsilon,
\eeas
where $J_m(t), m\in \mathbb{Z}$ denote Bessel functions of the first kind, and 
\[
r(|t|,\varepsilon) = O\left(|t| + \frac{\log(1/\varepsilon)}{\log(e+|t|^{-1}\log(1/\varepsilon))}  \right).
\]
\end{prop}
}

\bibliographystyle{plain}
\bibliography{main}

\end{document}